%% file: linearpool.tex
\documentclass[journal]{IEEEtran}
\usepackage{amssymb,amsmath,amsthm}
\usepackage{graphicx}
\usepackage{booktabs}
\usepackage{subcaption}
\usepackage[ruled,vlined,linesnumbered,lined,boxed]{algorithm2e}

\usepackage[backend=biber,style=ieee]{biblatex}
\newcommand{\SSCMk}{{\mathbf{S}_{\text{sgn},k}}}

\newcommand{\SSCMpop}{{\boldsymbol \Sigma}_\text{sgn}}
\newcommand{\SSCMshapepop}{{\boldsymbol \Lambda}_\text{sgn}}
\newcommand{\SSCMshapepopk}{{\boldsymbol \Lambda}_{\text{sgn},k}}
\newcommand{\SSCMshape}{{\hat{\boldsymbol \Lambda}}}
\newcommand{\shape}{{\boldsymbol \Lambda}}
\newcommand\M{{\boldsymbol \Sigma}}
\renewcommand\S{{\mathbf S}}
\newcommand{\LINPOOLI}{\tilde{\mathbf{S}}}
\newcommand{\LINPOOL}{{\mathbf{S}}}
\newcommand{\aLB}{\epsilon} % identity lowerbound
\newcommand{\aLBk}{\epsilon_k} % identity lowerbound for class k 
\newcommand{\lambdamax}{\lambda_{\max}}
\usepackage{textcomp}
\newcommand{\MATLAB}{{MATLAB\textsuperscript{\tiny\textregistered}}~}

\renewcommand\u{{\mathbf u}}
\newcommand\e{{\mathbf e}}
\newcommand\w{{\mathbf w}}
\newcommand\X{{\mathcal X}}
\newcommand\x{{\mathbf x}}
\newcommand\m{{\mathbf m}}
\newcommand\bmu{{\boldsymbol \mu}}
\newcommand\I{{\mathbf I}}
\newcommand\T{{\mathbf T}}

\newcommand\A{{\mathbf A}}
\renewcommand\a{{\mathbf a}}
\newcommand\B{{\mathbf B}}
\newcommand\C{{\mathbf C}}
\newcommand\cv{{\mathbf c}}
\newcommand\Q{{\mathbf Q}}
\newcommand\f{{\boldsymbol{\eta}}}
\newcommand\D{{\boldsymbol\Delta}}
\newcommand\one{{\boldsymbol 1}}
\newcommand\zero{{\boldsymbol 0}}
\newcommand\real{{\mathbb{R}}} % real numbers
\newcommand\complex{{\mathbb{C}}} % complex numbers
\newcommand\Fro{{\mathrm{F}}}
\DeclareMathOperator{\E}{\mathbb{E}}
\DeclareMathOperator{\var}{\mathrm{var}}

\DeclareMathOperator{\tr}{tr}
\DeclareMathOperator{\diag}{diag}
\renewcommand{\vec}{\mathrm{vec}}

\newcommand{\argmin}{\operatornamewithlimits{arg~min\ }}

\newcommand{\hop}{\mathsf{H}}
\newcommand{\NMSE}{\mathrm{NMSE}}

\usepackage{pgfplots}
\pgfplotsset{compat=1.13}
\usepgfplotslibrary{statistics}
\usetikzlibrary{calc}
\usepgfplotslibrary{groupplots}

\newcommand\numberthis{\addtocounter{equation}{1}\tag{\theequation}}

\usepackage{color}

\newcommand\blue[1]{{#1}}

\usepackage{cancel}

\theoremstyle{remark}
\newcounter{ctheorem}
\newtheorem{theorem}[ctheorem]{Theorem}

\newcounter{cremark}
\newtheorem{remark}[cremark]{Remark}

\newcounter{clemma}
\newtheorem{lemma}[clemma]{Lemma}
\newcounter{cproposition}
\newtheorem{proposition}[cproposition]{Proposition}

\newcommand{{\papertitle}}{Linear pooling of sample covariance matrices}
\newcommand{{\keywords}}{covariance matrix, elliptical distribution,
high-dimensional, multiclass, regularization, shrinkage, spatial sign
covariance matrix}

\usepackage[pdftex,
pdfauthor={Elias Raninen, David E. Tyler, and Esa Ollila},
pdftitle={\papertitle},
pdfkeywords={\keywords},
pdfcreator={pdflatex}]{hyperref}

\begin{document}
\title{\papertitle}
\author{Elias~Raninen,~\IEEEmembership{Student~Member,~IEEE,}
David~E.~Tyler,
Esa~Ollila,~\IEEEmembership{Senior~Member,~IEEE}%
\thanks{E. Raninen and E. Ollila are with the Department of Signal
Processing and Acoustics, Aalto University, P.O. Box 15400, FI-00076
Aalto, Finland. D. E. Tyler is with the Department of Statistics, Rutgers, The
State University of New Jersey, Piscataway NJ 08854, U.S.A. 

The research work by E. Raninen and E. Ollila was supported in part by the
Academy of Finland under Grant 298118. Research work by D. E. Tyler was
supported in part by the National Science Foundation under Grant DMS-1812198.}}%

\maketitle
\begin{abstract} % 100 to 150 words
    We consider the problem of estimating high-dimensional covariance matrices
    of $K$-populations or classes in the setting where the sample sizes are
    comparable to the data dimension.  We propose estimating each class
    covariance matrix as a distinct linear combination of all class sample
    covariance matrices. This approach is shown to reduce the estimation error
    when the sample sizes are limited, and the true class covariance matrices
    share a somewhat similar structure. We develop an effective method for
    estimating the coefficients in the linear combination that minimize the mean
    squared error under the general assumption that the samples are drawn from
    (unspecified) elliptically symmetric distributions possessing finite
    fourth-order moments. To this end, we utilize the spatial sign covariance
    matrix, which we show (under rather general conditions) to be an
    asymptotically unbiased estimator of the normalized covariance matrix as the
    dimension grows to infinity. We also show how the proposed method can be
    used in choosing the regularization parameters for multiple target matrices
    in a single class covariance matrix estimation problem. We assess the
    proposed method via numerical simulation studies including an application in
    global minimum variance portfolio optimization using real stock data.
\end{abstract}
\begin{IEEEkeywords}
    \keywords
\end{IEEEkeywords}
\IEEEpeerreviewmaketitle

\section{Introduction}
\label{sec:intro}

\IEEEPARstart{H}{igh-dimensional} covariance matrix estimation is a challenging
problem as the dimension $p$ of the observations can be much larger than the
sample size $n$.  Such problems are increasingly common, for example, in
finance~\cite{ledoitImprovedEstimationCovariance2003},
genomics~\cite{schaferShrinkageApproachLargeScale2005}, graphical
models~\cite{zhangMultivariateGeneralizedGaussian2013},
\blue{chemometrics~\cite{engelOverviewLargedimensionalCovariance2017}}, wireless
sensor networks~\cite{yilunchenRobustShrinkageEstimation2011}, and adaptive
filtering in array signal
processing~\cite{bessonRegularizedCovarianceMatrix2013}. This paper considers a
high-dimensional problem, where there are $K$ distinct classes (populations).
Since the population variables are in general the same, but are measured under
different population conditions, it is reasonable to presume the $K$ distinct
covariance matrices share some common features or structure.  In the small
sample size setting, it is then advantageous to leverage on this presumption in
the estimation of the population covariance matrices.

Consider $K$ mutually independent classes, where each class $k \in \{1,\ldots,K\}
$ consists of independent and identically distributed (i.i.d.) $p$-dimensional
observations $\X_k = \{\x_{1,k}, \ldots, \x_{n_k,k} \}$ of size $n_k$ with mean
$\bmu_k = \E[\x_{i,k}]$ and positive definite symmetric covariance matrix
\begin{equation}\label{eq:Sigma}
    \M_k = \E[(\x_{i,k} - \bmu_k)(\x_{i,k} - \bmu_k)^\top].
\end{equation}
The ordinary estimators for the covariance matrix and the mean are the sample
covariance matrix (SCM)
\begin{equation}\label{eq:SCM}
    \S_k = \frac{1}{n_k-1} \sum_{i=1}^{n_k}
    (\x_{i,k} - \bar \x_k)(\x_{i,k} - \bar \x_k)^\top,
\end{equation}
and the sample mean $\bar \x_k = \frac{1}{n_k} \sum_{i=1}^{n_k} \x_{i,k}$.  When
the sample size $n_k$ and the data dimensionality $p$ are comparable in size,
the SCM can be highly variable, resulting in an unstable estimate of the
population covariance matrix. Also, the SCM is positive definite only if $n_k >
p$ and $\X_k$ spans $\real^p$. Due to these problems, a commonly used approach
in high-dimensional covariance matrix estimation is to use regularization
(shrinkage).

In the one population case ($K=1$), linear regularization usually refers to
estimating the covariance matrix by a linear or convex combination of the SCM
(or some other primary estimator) and a (usually positive definite) target
matrix. Multiple examples can be found in
finance~\cite{ledoitImprovedEstimationCovariance2003,
ledoitWellconditionedEstimatorLargedimensional2004,
ledoitHoneyShrunkSample2004,ollilaOptimalShrinkageCovariance2019},
genomics~\cite{schaferShrinkageApproachLargeScale2005}, and signal
processing~\cite{chenShrinkageAlgorithmsMMSE2010,
liEstimationLargeCovariance2018,steinerFastConvergingAdaptive2000,
abramovichRegularizedCovarianceMatrix2013,
aubryGeometricApproachCovariance2018,demaioLoadingFactorEstimation2019}.
The target matrix is chosen based on prior knowledge or assumptions about the
true population covariance matrix. Sometimes using more than one target matrix
can further reduce the estimation error. In the double shrinkage approach
of~\cite{ikedaComparisonLinearShrinkage2016}
and~\cite{halbeRegularizedMixtureDensity2013} there are two convex
regularization steps: the SCM is first regularized toward a diagonal matrix
consisting of the sample variances after which the resulting estimator is
further regularized toward a scaled identity matrix. Recently, also
multi-target shrinkage methods have been proposed that are able to incorporate a
larger number of simultaneous target
matrices~\cite{lancewickiMultitargetShrinkageEstimation2014,
bartzMultitargetShrinkage2014, tongLinearShrinkageEstimation2018,
riedlMultimodelShrinkageKnowledgeaided2018, grayShrinkageEstimationLarge2018,
zhangImprovedCovarianceMatrix2019}. 

In the multiple population setting, regularization via pooling the information
in the different class samples is also possible. For
example,~\cite{bessonMaximumLikelihoodCovariance2020} considered covariance
matrix estimation from two independent data sets, whose covariance matrices are
different but close to each other. This type of problems are encountered in
radar processing as well as in hyperspectral imaging applications, where
additional data sets may have been acquired with slightly different measurement
configurations~\cite{bessonMaximumLikelihoodCovariance2020}. In the context of
wireless sensor networks,~\cite{colucciaRobustOpportunisticInference2016}
considered linear parameter estimation from independent and non-identically
distributed scalar sample statistics. In discriminant analysis classification,
the pooled SCM, $\S_{\text{pool}} = \frac{1}{n}\sum_{k=1}^K n_k \S_k$, $n =
\sum_{k=1}^K n_k$, is often used as a shrinkage target and the class covariance
matrices are estimated via a convex combination $\hat \M_k = a \S_k + (1-a)
\S_{\text{pool}}$, where $a \in [0,1]$. This was studied in a Bayesian framework
in~\cite{greenePartiallyPooledCovariance1989}
and~\cite{rayensCovariancePoolingStabilization1991}, and in the
\emph{Regularized Discriminant Analysis} (RDA) framework
in~\cite{friedmanRegularizedDiscriminantAnalysis1989}. Furthermore, the optimal
tuning parameter for this setting under elliptically distributed data was
derived in our earlier work~\cite{raninenOptimalPoolingCovariance2018,
raninenCoupledRegularizedSample2021}. For applications using RDA see,
e.g.,~\cite{krasoulisUseRegularizedDiscriminant2017,
luRegularizedDiscriminantAnalysis2003}

As noted previously, at least some of the $K$ population covariance matrices can
be similar (close to each other in terms of suitable distance metric) and so
it would be beneficial to use regularization to reduce the variance of the final
estimates of the covariance matrices. Following this idea, we propose to
estimate each class covariance matrix as a nonnegative linear combination of the
SCMs of all classes. For $\a \geq \zero$, define
\begin{equation}
    \LINPOOL(\a) = \sum_{i=1}^K a_{i} \S_i.
    \label{eq:linearcombination}
\end{equation}
Restricting the coefficients to be nonnegative ensures that the
estimator~\eqref{eq:linearcombination} is positive semidefinite.  The aim is to
find a $K \times K$ nonnegative coefficient matrix $\A = (a_{ij}) =
\begin{pmatrix}
    \a_1 & \cdots & \a_K
\end{pmatrix} \geq \zero$,
that minimizes the total mean squared error (MSE),
\begin{align}
    \A^\star &= \argmin_{\A \geq \zero}
    \sum_{k=1}^K \E \big[ \|\LINPOOL(\a_k) - \M_k \|^2_\Fro \big]
    \label{eq:problem}
    \\
    \Leftrightarrow
    \a_k^\star &= \argmin_{\a \geq \zero}
    \E \big[ \|\LINPOOL(\a) - \M_k \|^2_\Fro \big], ~ k=1,\ldots,K,
    \label{eq:ak}
\end{align}
with the estimate of $\M_k$ then taken to be $\hat \M_k = \LINPOOL(\a_k^\star)$.
The equivalence of~\eqref{eq:problem} and~\eqref{eq:ak} is evident from the fact
that the optimization problems for each class are separable. The solution to
this problem is given in Section~\ref{sec:pooledestimator}. It is easy to see
that the RDA based estimators form a subset of the more general
form~\eqref{eq:linearcombination}, which permits using individual weights for
each SCM in the sum. 

Below we summarize the main contributions of the paper.
\begin{itemize}
    \item We propose covariance matrix estimators for multiclass problems,
        based on linearly pooling the class SCMs. Several aspects and properties
        of the estimator are discussed including possible modifications and
        an extension for complex-valued data.
    \item We show how the optimal linear coefficients can be estimated by
        assuming that the data is elliptically distributed. To this end, we use
        the spatial sign covariance matrix (SSCM), which we show under rather
        general assumptions to be asymptotically unbiased with respect to
        growing dimension.
    \item We show how the estimator can be used as a multi-target shrinkage
        estimator in a single class problem.
    \item Numerical simulations are conducted including a portfolio optimization
        problem using real stock data. The simulations show promising
        performance of the proposed estimator compared to competing estimators.
    \item Code is available at
        \href{https://github.com/EliasRaninen}{github.com/EliasRaninen}, which
        works both for real and complex-valued data sets.
\end{itemize}

The rest of the paper is organized as follows. In
Section~\ref{sec:pooledestimator}, we derive the optimal coefficients for the
linear pooling estimator and study some of its properties. In
Section~\ref{sec:estimation}, we propose methods for estimating the statistical
parameters needed for the estimation of the optimal coefficients. This section
also presents the results regarding the SSCM. In Section~\ref{sec:extensions},
we discuss possible extensions and modifications to the estimator. In
Section~\ref{sec:multitarget}, we discuss the similarities and differences
between the proposed method and closely related multi-target shrinkage
covariance matrix estimation methods. Furthermore, we show how our proposed
method can be used as a multi-target shrinkage covariance matrix estimator in a
single class problem with arbitrary positive semidefinite target matrices.
Section~\ref{sec:simulations} provides numerical simulation studies and
Section~\ref{sec:portfolio} provides an application to investment portfolio
selection using historical stock data.  Lastly, Section~\ref{sec:conclusions}
concludes.

\emph{Notation:}
Matrices are denoted by upper case boldface letters ($\A$ or $\D$), vectors are
denoted by lower case boldface letters ($\a$ or $\boldsymbol \delta$), and
scalars are denoted by lower case letters ($a$ or $\delta$). For a matrix $\A =
(a_{ij})$, the notation $\A \geq \zero$ means that the matrix is nonnegative,
that is, $a_{ij} \geq 0$, for all $i$ and $j$.  Similarly, for a vector $\a =
(a_i)$ the notation $\a \geq \zero$ means that $a_i \geq 0$, for all $i$. The
notation $\A \succ \zero$ ($\A \succeq \zero)$ means that $\A$ is positive
definite (positive semidefinite). The notation $\diag(\a)$ denotes a diagonal
matrix with the entries of $\a$ on the main diagonal. The identity matrix is
denoted by $\I$ and the vector of all ones is denoted by $\one$.  The notation
$\e_i$ denotes the $i$th Euclidean basis vector, i.e., a vector whose $i$th
coordinate is 1 and all other coordinates are 0. For real sequences $a_p$ and
$b_p$, as $p \to \infty$, the notation $a_p = o(b_p)$ means that the sequence
$a_p/b_p \to 0$, and the notation $a_p = O(b_p)$ means that the sequence
$a_p/b_p$ is bounded.  For a matrix-valued sequence $\A_p$, we write $\A_p =
o(b_p)$ and $\A_p = O(b_p)$ if and only if $\|\A_p\|_\Fro = o(b_p)$ and
$\|\A_p\|_\Fro = O(b_p)$, respectively. The Frobenius norm is defined as
$\|\A\|_\Fro = \sqrt{\tr(\A^\top \A)}$ while $\| \cdot \|$ denotes the
Euclidean norm for vectors. The largest eigenvalue of $\A$ is
$\lambdamax(\A)$. The determinant of $\A$ is denoted by $|\A|$.

\section{Linear pooling of SCMs} \label{sec:pooledestimator}
In this section, we address solving for the coefficients of the linear
combination of SCMs in \eqref{eq:linearcombination}. First, define the
\emph{scaled MSE} of the SCM $\S_k$ as
\begin{equation*}
    \delta_k
    = p^{-1}\mathrm{MSE}(\S_k) = p^{-1}\E[ \| \S_k - \M_k \|_\Fro^2]
\end{equation*}
and the scaled inner products of the covariance matrices as
\begin{equation*}
    c_{ij} = p^{-1} \tr(\M_i \M_j),
\end{equation*}
\blue{where $p$ is the dimension of the data.}
Then define
\begin{equation*}
    \D = \diag(\delta_1,\ldots,\delta_K)
    ~\text{and}~
    \C = \begin{pmatrix} 
            \cv_1 \cdots \cv_K 
         \end{pmatrix} 
         = (c_{ij}).
\end{equation*}
We can then state the following result.

\begin{theorem}\label{thm:MSE}
    (\emph{The MSE of the linearly pooled estimator}) For class $k$, the MSE
    in \eqref{eq:ak} can be written as the strictly convex quadratic function
    \begin{equation*}
        p ( \a^\top (\D + \C) \a - 2 \cv_k^\top \a + c_{kk} ),
    \end{equation*}
    where $\D + \C$ is a positive definite symmetric matrix.
\end{theorem}
\begin{proof}
    See Appendix \ref{app:thm:MSE}.
\end{proof}
As a consequence of Theorem \ref{thm:MSE}, the optimal coefficients can be
computed in the following way.
\begin{proposition}
    (\emph{Optimal nonnegative coefficients})
    The solution to \eqref{eq:ak}, $\a_k^\star$, is found by solving the
    strictly convex quadratic programming (QP) problem
    \begin{equation} \label{eq:AQP}
	\begin{array}{ll}
	    {\text{minimize}}
	    & \frac{1}{2} \a^\top (\D + \C) \a - \cv_k^\top \a \\
	    \text{subject to} & \a \geq \zero.
	\end{array}
    \end{equation}
\end{proposition}
\begin{proof}
    Follows from Theorem~\ref{thm:MSE}.
\end{proof}
Many efficient algorithms exist for solving constrained convex
QPs~\cite{nocedalNumericalOptimization2006}.\footnote{In~\cite{kozlovPolynomialSolvabilityConvex1980}
    it was shown that using the ellipsoidal method, the strictly convex QP can
    be solved in polynomial time. In the simulations, we use the \MATLAB
function \texttt{quadprog}, which uses an interior-point method.} The
optimization problem~\eqref{eq:AQP} requires knowledge of the matrices $\C$ and
$\D$, which depend on the unknown population parameters. We can nonetheless estimate
the solution by using estimates $\hat \C$ and $\hat \D$, which can be computed
from the data as explained in Section~\ref{sec:estimation}.

It is instructive to consider an unconstrained version of the optimization
problem~\eqref{eq:AQP}, where the weights are allowed to take negative values.
For this case, we have the following closed form solution.
\begin{proposition}\label{prop:unconstrained}
    (\emph{Optimal unconstrained coefficients})
    The unconstrained solution, which minimizes the MSE in \eqref{eq:problem}
    is
    \begin{align}
        \a_k^\star = (\D + \C)^{-1} \cv_k
        \Leftrightarrow
        \A^\star = (\D + \C)^{-1} \C.
        \label{eq:unconstrained}
    \end{align}
\end{proposition}
\begin{proof} 
    Follows from Theorem \ref{thm:MSE}.
\end{proof}
Note that if the closed form solution $\a_k^\star \geq \zero$
in~\eqref{eq:unconstrained}, then it is also the solution to~\eqref{eq:AQP}.

Consider the single class case, for which the problem reduces to finding an
optimal scaling parameter $a_1$ such that $\E \big[\| a_1 \S_1 - \M_1 \|^2_\Fro
\big]$ is minimized.  Proposition~\ref{prop:unconstrained} above then states
that the optimal parameter that minimizes the MSE is
\begin{equation*}
    a_1^\star = \frac{c_{11}}{\delta_1 + c_{11}} =
    \frac{1}{\mathrm{NMSE}(\S_1) + 1},
\end{equation*}
where $\mathrm{NMSE}(\S_1) = \mathrm{MSE}(\S_1)/\|\M_1\|_{\Fro}^2$ is referred to
as the normalized MSE (NMSE).  It can easily be shown that
$\mathrm{MSE}(a_1^\star \S_1) = a_1^\star\mathrm{MSE}(\S_1) <
\mathrm{MSE}(\S_1)$ since $a_1^\star < 1$. Therefore, the (oracle) estimator
$\hat \M_1 = a_1^\star \S_1$ is always more efficient than $\S_1$.  For the
univariate normal case, one obtains $a_1^\star = (n_1-1)/(n_1+1)$. This result
was first obtained in~\cite{goodmanSimpleMethodImproving1953}. A corresponding
result for the general (non-normal) univariate case was obtained
in~\cite{searlsNoteEstimatorVariance1990}, and it can be written as $a_1^\star =
((n_1+1)/(n_1-1) + 3\kappa_1/n_1)^{-1}$, where $\kappa_1$ is the symmetric
kurtosis of the population (see~\eqref{eq:kappa_k}).

Consider next the special case when all population covariance matrices are
equal, i.e., $\M_1 = \cdots = \M_K \equiv \M$. In this case, $\C = c\one
\one^\top$ with $c = \tr(\M^2)/p$. Using the Sherman-Morrison formula, $(\D +
c\one \one^\top)^{-1} = \D^{-1} - \mu \D^{-1} \one \one^\top \D^{-1}$, where
$\mu = c/(1+ c\one^\top \D^{-1} \one)$, we obtain the
solution~\eqref{eq:unconstrained} as
\begin{equation*}
    \A^\star = \mu \D^{-1}\one \one^\top.
\end{equation*}
Hence, all columns of $\A^\star$ coincide, and the coefficients in each column
are $a_{jk}^\star = \mu/\delta_j$. These coefficients can be written in an
equivalent form
\begin{equation*}
    a_{jk}^\star = \frac{ \mathrm{NMSE}(\S_j)^{-1}}{1 + \sum_{i=1}^K
    \mathrm{NMSE}(\S_i)^{-1}}.
\end{equation*}
Thus, the weights are positive and proportional to the inverses of the NMSE of
the SCMs. If $\S_j$ has a large NMSE relative to others, which occurs for
example when the sample size $n_j$ is small relative to others, then the
weight $a_{jk}^\star$ assigned for $\S_j$ is small. This implies that the
contribution of $\S_j$ in the linear combination $\LINPOOL(\a_k^\star)$ is
small.

If one further assumes that all populations have the same distribution and the
sample sizes are equal, then $\delta_1 = \cdots = \delta_K \equiv \delta$ and
\begin{equation*}
    a^\star \equiv  a_{jk}^\star = \frac{1}{\mathrm{NMSE}(\S_j) + K} <
    \frac{1}{K} \quad \forall j,k.
\end{equation*}
That is, the pooled SCM is shrunk by the factor $K a^\star < 1$.

Due to the positivity of the coefficients, the conclusions of these special
cases naturally also hold for the constrained case, where the coefficients are
constrained to be nonnegative.

\section{Estimation}\label{sec:estimation}
In this section, we address the estimation of $\D$ and $\C$. We review
elliptically symmetric distributions, introduce the relevant statistical
parameters as well as show how to estimate them. Regarding an estimate for the
sphericity parameter, we use the SSCM for which we then prove an asymptotic
unbiasedness result in Theorem~\ref{thm:sign}.

\subsection{Elliptically symmetric distributions}
We will assume that the samples are generated from unspecified elliptically
symmetric distributions with finite fourth-order moments. That is, an absolutely
continuous random vector $\x \in \real^p$ from the $k$th population is assumed
to have a density function up to a constant of the form
\begin{equation*}
    |\M_k|^{-1/2} g_k ( (\x - \bmu_k)^\top \M_k^{-1} (\x - \bmu_k) ),
\end{equation*}
where $g_k:\real_{\geq 0} \to \real_{> 0}$ is called the \emph{density
generator}~\cite{bilodeauTheoryMultivariateStatistics1999}. Here, we let $g_k$
to be defined so that $\M_k$ represents the covariance matrix of $\x$, which
means that $C_k^{-1} \int_0^\infty t^{p/2} g_k(t) \mathrm{d} t = p$, where
$C_k = \int_0^\infty t^{p/2-1} g_k(t) \mathrm{d} t$. 
We write $\x \sim \mathcal E_p(\bmu_k,\M_k,g_k)$ to denote this case. For
example, the multivariate normal (MVN) distribution is a particular instance of
the elliptical distribution for which $g_k(t) = \exp(-t/2)$. We write $\x \sim
\mathcal N_p(\bmu_k,\M_k)$ to denote this case. An elliptically distributed
random vector $\x \sim \mathcal E_p(\bmu_k,\M_k,g_k)$ can be expressed by the
\emph{stochastic representation} as
\begin{equation}\label{eq:stochasticrepresentation}
    \x = \bmu_k + r_k \M_k^{1/2} \mathbf{u},
\end{equation}
where $r_k$ is a random variable called the \emph{modular variate},  verifying
$\E[r_k^2]=p$, and $\mathbf{u}$ is a random vector distributed uniformly on the
unit sphere, i.e., $\mathbf{u} \in \{\mathbf{z} \in \real^p :
\|\mathbf{z}\|=1\}$.  Furthermore, $\mathbf{u}$ and $r_k$ are independent.  More
generally, we note that any random vector $\x$ which
satisfies~\eqref{eq:stochasticrepresentation} is said to have an elliptical
distribution, even if it is not absolutely continuous, i.e., does not have a
density. The relationship between the modular variate $r_k$ and $\x$ is readily
seen from~\eqref{eq:stochasticrepresentation} to be $r_k^2 = (\x - \bmu_k)^\top
\M_k^{-1} (\x - \bmu_k)$.

Sometimes we are only interested in the covariance matrix up to a scaling
constant. Hence, we define the \emph{shape matrix}:
\begin{equation*}
    \shape_k = p \frac{\M_k}{\tr(\M_k)}, 
\end{equation*}
which verifies $\tr(\shape_k)=p$.  Additionally, we define three statistical
parameters that describe the elliptical distribution. First, we define the
\emph{scale}:
\begin{equation}\label{eq:eta_k}
    \eta_k = \tr(\M_k)/p,
\end{equation}
which is equal to the mean of the eigenvalues of $\M_k$. Note that, $\M_k =
\eta_k \shape_k$. Second, we define the \emph{sphericity}:
\begin{equation}\label{eq:gamma_k}
    \gamma_k = \frac{p \tr(\M^2_k)}{\tr(\M_k)^2} = \frac{\tr(\shape_k^2)}{p},
\end{equation}
which equals the ratio of the mean of the squared eigenvalues relative to the
mean of the eigenvalues squared. The sphericity parameter gets values in the
range $[1,p]$ and attains its minimum for the scaled identity matrix and its
maximum for a rank one matrix. Third, we define the \emph{elliptical kurtosis}:
\begin{equation} \label{eq:kappa_k}
    \kappa_k = \frac 1 3 \mathrm{kurt}(x_i) = \frac 13 \left(
    \frac{\E[(x_i-\mu_i)^4]}{\E[(x_i-\mu_i)^2]^2} - 3 \right),
\end{equation}
where $x_i$ denotes any marginal variable of $\x = (x_i) \sim \mathcal
E_p(\bmu_k,\M_k,g_k)$ and $\mu_i = \E[x_i]$. For example, if the sample is from
a MVN distribution, then $\kappa_k=0$. The kurtosis parameter also satisfies
$\kappa_k = \E[r_i^4]/(p(p+2))- 1 $, and hence, not only does $\kappa_k$
represent the kurtosis of each of the variables $x_i$, but it also represents
the kurtosis of any univariate linear combination $\mathbf{b}^\top \x$, where
$\mathbf{b} \in \real^p \setminus \{\zero\}$. The lower bound for the kurtosis
parameter is $\kappa^{\text{LB}} =
-2/(p+2)$~\cite{bentlerGreatestLowerBound1986}.

For elliptical populations, the scaled MSE $\delta_k$ of the SCM obtains an
explicit form \cite[Lemma~1]{ollilaOptimalShrinkageCovariance2019}:
\begin{equation}
    \delta_k = \eta_k^2 \Big(
    \Big(\frac{1}{n_k - 1} + \frac{\kappa_k}{n_k}\Big)(p + \gamma_k)
    + \frac{\kappa_k}{n_k} \gamma_k \Big),
    \label{eq:delta_k}
\end{equation}
which depends on the known sample size $n_k$ as well as the unknown scale
$\eta_k$ \eqref{eq:eta_k}, the unknown sphericity $\gamma_k$ \eqref{eq:gamma_k},
and the unknown elliptical kurtosis $\kappa_k$ \eqref{eq:kappa_k}.

\subsection{Estimates of the scale and the elliptical kurtosis}
We estimate the scale using the SCM via 
\begin{equation}\label{eq:etahat}
    \hat \eta_k = \tr(\S_k)/p.
\end{equation}
The kurtosis $\hat \kappa_k$ is estimated via the (bias-corrected) average
sample kurtosis of the marginal variables,
\begin{align*}\label{eq:kappahat}
    &
    \hat \kappa_k
    =
    \frac{N}{3} ( (n_k+1) g_2 + 6 ),
    \numberthis
    \\
    &
    g_2
    = 
    \frac{1}{p} \sum_{j=1}^p
    \frac{\frac{1}{n_k}\sum_{i=1}^{n_k} (x_{ij,k} - \bar x_{j,k})^4}
    {\left(\frac{1}{n_k}\sum_{i=1}^{n_k} (x_{ij,k}- \bar x_{j,k})^2\right)^2} - 3,
    \\
    &
    N = \frac{n_k-1}{(n_k-2)(n_k-3)}
    ,
\end{align*}
where $\bar x_{j,k} = \frac{1}{n_k}\sum_{i=1}^{n_k}
x_{ij,k}$~\cite{joanesComparingMeasuresSample1998}. In case \eqref{eq:kappahat}
is less than $\kappa^{\text{LB}}$, we set $\hat \kappa_k = 0.99
\kappa^{\text{LB}}$.  Note that, although $\kappa_k$ is invariant under affine
transformations of $\x$, the estimator $\hat \kappa_k$ in \eqref{eq:kappahat} is
not. An alternative estimator of  $\kappa_k$ that is affine equivariant is 
\[ 
    \tilde \kappa_k =
    \frac{1}{n_k} \sum_{i=1}^{n_k} \frac{ [  (\x_{i,k} - \bar \x_k)^\top
    \S_k^{-1} (\x_{i,k} - \bar \x_k)]^2}{p (p+2)}- 1. 
\] 
This estimator requires  $n_k >  p$, and hence we use $\hat \kappa_k$.

\subsection{Estimate of the sphericity using the SSCM}\label{sec:signestimators}
Regarding the sphericity, it would be natural to develop an estimator using the
SCM as well. However, a simple and particularly well performing estimator of the
sphericity is based on the robust \emph{spatial sign covariance matrix} (SSCM)
\eqref{eq:SSCM} and it has been used, e.g., in
\cite{zouMultivariateSignbasedHighdimensional2014, chengTestingEqualityTwo2019,
zhangAutomaticDiagonalLoading2016},
and~\cite{ollilaOptimalShrinkageCovariance2019}. Particularly,
in~\cite{ollilaOptimalShrinkageCovariance2019}, both a SCM and a SSCM based
estimator of the sphericity was compared and, except for the case where the
samples were MVN, the simulations suggested the superiority of the SSCM based
sphericity estimator. Before introducing the estimator for the sphericity, we
will discuss the properties of the SSCM.

For $\x \sim \mathcal E_p (\bmu, \M, g)$, we define the population SSCM as
\begin{equation}
    \SSCMpop = \E\Bigg[ \frac{(\x - \bmu)(\x - \bmu)^\top}{\|\x - \bmu
    \|^2}\Bigg],
\end{equation}
and since $\tr(\SSCMpop) = 1$, its corresponding shape matrix is given by
$\SSCMshapepop = p \SSCMpop$.

It is known that $\SSCMshapepop$ and $\shape$ have the same eigenvectors as well
as the multiplicities and the orders of the eigenvalues, but the eigenvalues
themselves are different~\cite{durreEigenvaluesSpatialSign2016}. The
\emph{sample SSCM} of the $k$th population and its corresponding shape estimator
are defined as 
\begin{equation}\label{eq:SSCM} 
    \SSCMk 
    = 
    \frac{1}{n_k} \sum_{i=1}^{n_k} 
    \frac{(\x_{i,k} - \bmu_k)(\x_{i,k} - \bmu_k)^\top}
    {\|\x_{i,k} - \bmu_k\|^2}
    ~\text{and}~\SSCMshape_k = p \SSCMk,
\end{equation} 
where the mean $\bmu_k$ is replaced with the \emph{sample spatial
median}~\cite{brownStatisticalUsesSpatial1983}, $\hat \bmu_k = \argmin_{\m}
\sum_{i=1}^{n_k}\|\x_{i,k} - \m\|$, when it is unknown. When the mean is known
$\E[\SSCMshape_k] = \SSCMshapepopk$, and $\SSCMshape_k$ is distribution-free
over the class of elliptical distributions. The latter statement can be proved
by writing~\eqref{eq:SSCM} in terms of the stochastic
representation~\eqref{eq:stochasticrepresentation} and observing that the
modular variate $r_k$ cancels out in each summand. 

Since the eigenvalues of $\SSCMshapepop$ and $\shape$ are different,
$\SSCMshape_k$ is biased~\cite{magyarAsymptoticInadmissibilitySpatial2014}.
Surprisingly, as shown in Theorem~\ref{thm:sign} below, this bias becomes
more negligible in higher dimensions provided the sequence of covariance
matrix structures being considered with increasing $p$ satisfies  
\begin{equation} \label{eq:gamma_oh_p}
    \gamma = o(p) \mbox{ as $p\to \infty$}.
\end{equation}
\begin{theorem}\label{thm:sign}
    Let $\x \sim \mathcal E_p(\bmu,\M,g)$. Then,
    \begin{align*}
        \SSCMshapepop = \shape + O(\gamma).
    \end{align*}
    Furthermore, if \eqref{eq:gamma_oh_p} holds, then $\SSCMshapepop = \shape +
    o(\|\shape\|_\Fro)$.
\end{theorem}
\begin{proof}
    See Appendix \ref{app:thm:sign}.
\end{proof}

The central condition~\eqref{eq:gamma_oh_p} holds for many common sequences of
covariance models, as shown in Proposition~\ref{prop:sphericity} below.
We first present, in the following lemma, a simple general condition under
which~\eqref{eq:gamma_oh_p} holds. This lemma is seen to hold in particular for
the case when the eigenvalues of $\shape$ are bounded as $p \to \infty$.

\begin{lemma}\label{lemma:lambda1}
    If $\shape$ is a shape matrix for which $\lambdamax(\shape)= O(p^{\tau/2})$,
    where $\tau < 1$, as $p \to \infty$, then \eqref{eq:gamma_oh_p} holds. 
\end{lemma}
\begin{proof}
    See Appendix~\ref{app:theoreticalsphericity}.
\end{proof}

\begin{proposition}\label{prop:sphericity}
    The following sequences, in $p$, of covariance models
    satisfy~\eqref{eq:gamma_oh_p}. 
    \begin{itemize}
        \item \emph{First order autoregressive} (AR(1)) covariance matrices:
            $(\M)_{ij} = \sigma^2 \varrho^{|i-j|}$, where $|\varrho| < 1$ and
            $\sigma^2 > 0$ are both fixed, i.e., they do not depend on the
            dimension $p$.

        \item \emph{1-banded Toeplitz covariance matrices}: $(\M)_{ij} =
            \sigma^2 \varrho^{|i-j|}$ for $|i-j| \leq 1$ and $(\M)_{ij} = 0$
            otherwise, where $|\varrho| < -(2\cos(p\pi/(p+1) )^{-1}$ and
            $\sigma^2 > 0$ are both fixed.

        \item \emph{Spiked covariance matrices}: $\M = \M_r + \alpha \I$, where
            $\M_r$ is positive semidefinite with rank $r \leq p$ and $ (r/p)
            [\lambda_{\max}(\M_r) /\alpha]^2 = o(p)$. Here, $\alpha$, $r$, and
            $\lambda_{\max}(\M_r)$ may vary with $p$.
    \end{itemize}
     However, for the \emph{compound symmetric} (CS) covariance matrix
     $(\M)_{ij} = \sigma^2 \varrho$, for $i\neq j$ and $(\M)_{ii} = \sigma^2$,
     where $\varrho \in (-(p-1)^{-1}, 1)$ and $\sigma^2 > 0$ are both fixed, one
     obtains $\gamma = O(p)$.
\end{proposition}
\begin{proof}
    See Appendix~\ref{app:theoreticalsphericity}.
\end{proof}

The restrictions on $\varrho$ for the covariance models in
Proposition~\ref{prop:sphericity} are needed to ensure that $\M$ is positive
definite.

Let us illustrate this result in the case that $\M$ has an AR(1) covariance
structure. In this case $\gamma \to (1+\varrho^2)/(1-\varrho^2)$ (see
Appendix~\ref{app:AR1}). For $\varrho=0.5$, $\gamma \to 5/3 = O(1)$. From
Theorem~\ref{thm:sign}, we then have that the relative error
$\|\shape\|_\Fro^{-1}\|\SSCMshapepop - \shape\|_\Fro$ is of the order
$O(\gamma/\|\shape\|_\Fro) = O(\sqrt{\gamma/p}) = O(p^{-1/2})$.

\blue{
We also have the following proposition about the normalized mean squared
distance between $\SSCMshape$ and a scaled SCM.
\begin{proposition}\label{prop:SSCMminusSCM}
    Let $\X = (\x_i) \sim \mathcal N_p(\zero,\M)$ and assume
    that~\eqref{eq:gamma_oh_p} holds. Then,
    \begin{equation*}
        \frac{\E \big[ \|\SSCMshape - \shape_{\text{SCM}}\|_\Fro^2 \big]}
        {\|\shape\|_\Fro^2}
        \overset{p\to \infty}\longrightarrow
        \frac{2}{n},
    \end{equation*}
    where $\shape_{\text{SCM}} = \eta^{-1} \frac{1}{n}\sum_{i=1}^n
    \x_i\x_i^\top$ and $\eta = \tr(\M)/p$, so that $\E[\shape_{\text{SCM}}] =
    \shape$.
\begin{proof}
    See Appendix~\ref{app:SSCMminusSCM}.
\end{proof}
\end{proposition}
}

We may now focus on the sphericity estimator in~\eqref{eq:gammahat}. It can be
shown by straightforward calculation (see Appendix~\ref{app:SSCMminusSCM}) that
\begin{equation}
    \frac{\E[\tr(\SSCMshape^2)]}{p}
	=
	\frac{p}{n}
	+
	\frac{n-1}{n}
	\frac{\tr(\E[\SSCMshape]^2)}{p}.
	\label{eq:EtrtildeS2}
\end{equation}
If \eqref{eq:gamma_oh_p} holds, then by Theorem~\ref{thm:sign}, one has that 
\begin{equation} \label{eq:EtrtildeS2_apu}
    \tr(\E[\SSCMshape]^2)/p \to \gamma ~\text{as}~ p \to \infty.
\end{equation}
By \eqref{eq:EtrtildeS2} and \eqref{eq:EtrtildeS2_apu}, a natural estimator for
the sphericity of the $k$th class  is then
\begin{equation}\label{eq:gammahat}
    \hat \gamma_k = \frac{n_k}{n_k-1}\left( \frac{\tr(\SSCMshape_k^2)}{p} -
    \frac{p}{n_k}\right).
\end{equation}

In a high-dimensional setting, using the spatial median $\hat \bmu$
in~\eqref{eq:SSCM} results in nonnegligible error in the sphericity estimate.
This was shown in~\cite{zouMultivariateSignbasedHighdimensional2014}, which
considered SSCM based hypothesis testing of sphericity of the covariance matrix
(i.e., $H_0: \M \propto \I$), and where they used a similar sphericity statistic
as in~\eqref{eq:gammahat}. They also proposed a method for estimating and
correcting for this error. Hence, we use the corrected estimator of the
sphericity~\cite{zouMultivariateSignbasedHighdimensional2014}: 
\begin{equation}
    \hat \gamma_k^* = \hat \gamma_k - p d_k,
    \label{eq:gammahatstar}
\end{equation}%
where
\begin{align*}
    d_k &= \frac{1}{n_k^2} \cdot
    \Big(2 - 2\frac{q_{2,k}}{q_{1,k}^2} + \Big(\frac{q_{2,k}}{q_{1,k}^2}\Big)^2
    \Big)
    \\&
    + \frac{1}{n_k^3} \cdot
    \Big( 8 \frac{q_{2,k}}{q_{1,k}^2} - 6 \Big(\frac{q_{2,k}}{q_{1,k}^2}\Big)^2
    + 2 \frac{q_{2,k} q_{3,k}}{q_{1,k}^5} - 2 \frac{q_{3,k}}{q_{1,k}^3} \Big)
\end{align*}
and $q_{m,k} = (1/n_k) \sum_{i=1}^{n_k} \|\x_{i,k} - \hat \bmu_k \|^{-m}$. When
computational simplicity is desired, it is also possible to use $d_k \approx
n_k^{-2} + 2n_k^{-3}$, which is often a good approximation
(see~\cite{zouMultivariateSignbasedHighdimensional2014}).

\subsection{Final estimates}\label{sec:DhatChat}
The inner products $c_ {ij} = p^{-1} \tr(\M_i \M_j)$, for $i \neq j$, of the
matrix $\C$ are also estimated using SSCMs. In this case, however, no error
correction due to using the spatial median is needed
(see~\cite{chengTestingEqualityTwo2019}). The estimates for $\C$ and $\D$ are
thus $\hat \C = (\hat c_{ij})$ and $\hat \D =
\diag(\hat\delta_1,\ldots,\hat\delta_K)$, where 
\begin{equation}\label{eq:estimatorcrossterms}
    \hat c_{ij} =\begin{cases}
        \tr((\hat \eta_i \SSCMshape_i)(\hat \eta_j  \SSCMshape_j))/p,
        & \text{for $i \neq j$},
        \\
        \hat \gamma_i^* \hat\eta_i^2, & \text{for $i = j$},
    \end{cases}
\end{equation}
and $\hat\delta_k$ is obtained from~\eqref{eq:delta_k} via $\hat
\eta_k$~\eqref{eq:etahat},~$\hat \kappa_k$~\eqref{eq:kappahat}, and $\hat
\gamma_k^*$~\eqref{eq:gammahatstar}.

\section{Extensions and modifications}\label{sec:extensions}
In this section, we first show how to incorporate regularization towards the
identity matrix in the estimator. Then, we show how the optimal coefficients can
alternatively be solved via a semidefinite optimization problem, which enables
relaxing the nonnegativity constraint. Lastly, we extend the estimator to
complex elliptically symmetric distributions.

\subsection{Additional regularization towards the identity
matrix}\label{ssec:identity}
It is often beneficial to incorporate regularization towards the identity
matrix. For example, if  $p > n = \sum_k n_k$, then all of the SCMs are
singular. Regularization towards the identity can easily be added by using the
estimator
\begin{equation}\label{eq:LINPOOL}
    \LINPOOLI(\a) = \sum_{j=1}^K a_{j} \S_j + a_{I} \I,
\end{equation}
where $\a = (a_1,\ldots,a_K,a_I)^\top \in \real^{K+1}$. By constraining
$a_i~\geq~0$, $1 \leq j \leq K$, and $a_{I} \geq  \aLB$, where $\aLB>0$ is a
chosen lower bound for the identity regularization, the
estimator~\eqref{eq:LINPOOL} will be positive definite. Then
\begin{equation*}
    \a_k^\star = \argmin_{ a_j \geq 0,~ a_{I} \geq \aLBk} 
    \E \big[ \| \LINPOOLI(\a) - \M_k \|^2_\Fro \big]
\end{equation*}
can be solved via the strictly convex QP problem
\begin{equation} \label{eq:Linpool_QP}
      \begin{array}{ll}
            {\text{minimize}} &
            \frac{1}{2} \a^\top (\tilde \D + \tilde \C) \a - \tilde \cv_k^\top
                \a \\
                \text{subject to} &
                a_j \geq 0,~1 \leq j \leq K,  a_{I} \geq \aLBk
                ,
            \end{array}
\end{equation}
where
\begin{equation}\label{eq:tildeCD}
    \tilde\D = \diag(\delta_1,\cdots,\delta_K,0)
    \quad\text{and}\quad
    \tilde \C  =
        \begin{pmatrix}
            \C & \f \\
            \f^\top & 1
        \end{pmatrix},
\end{equation}
and $\f = (\eta_k) \in \real^K$ is a vector consisting of \emph{scales}
\eqref{eq:eta_k}. The unconstrained optimal solution for this case is $\A^\star
= (\tilde \D + \tilde \C)^{-1} \begin{pmatrix} \C & \f  \end{pmatrix}^\top$. The
positive definiteness of $\tilde \D + \tilde \C$ is shown in
Appendix~\ref{app:thm:MSE}. Algorithm~\ref{alg:linpool} summarizes the
procedure for computing the linearly pooled estimates of the class covariance
matrices.

\begin{algorithm}[!h]
    \caption{Linear pooling of SCMs with identity}\label{alg:linpool}
    \SetKwInOut{Input}{input}\SetKwInOut{Output}{output}
    \BlankLine
    \Input{Data $\mathcal X_1, \ldots, \mathcal X_K$ of the classes and
    $\aLBk > 0$, $k=1,\ldots,K$.}
	\Output{$\hat \M_1, \ldots, \hat \M_K$.}

    \BlankLine
    Compute SCMs $\S_1,\ldots,\S_K$ of the classes.

    Compute $\tilde\C$ and $\tilde\D$ of \eqref{eq:tildeCD} (estimate as in Section~\ref{sec:estimation}).

    Compute $\hat \A = (\hat \a_1 \cdots \hat \a_K) = (\tilde \D +
    \tilde\C)^{-1} \begin{pmatrix} \C & \f  \end{pmatrix}^\top$.

    \For{$k=1$ \KwTo $K$}{
        \If{$\exists i  \in \{1,\ldots,K\}: (\hat \a_k)_i < 0$ or $\hat a_{Ik} <  \aLBk$}
        {Set $\hat \a_k^\star$ as the solution of \eqref{eq:Linpool_QP}.
   }
        \Else{$\hat \a_k^\star \gets \hat \a_k$.}

        $\hat \M_k \gets \LINPOOLI(\hat \a_k^\star)$ of \eqref{eq:LINPOOL}.
	}
\end{algorithm}

\begin{remark}\label{remark:convexpooling}
    The QP formulation of the problem makes it easy to incorporate additional
    constraints if needed. For example, in order to find a convex combination of
    the SCMs the equality constraint $\one^\top \a = 1$ should be added to the
    QP~\eqref{eq:Linpool_QP}.
\end{remark} 

\subsection{Semidefinite programming formulation}
Constraining the coefficients in~\eqref{eq:AQP} to be nonnegative is sufficient
to ensure positive semidefiniteness of the
estimator~\eqref{eq:linearcombination}. In some cases this approach can be
suboptimal since it may be possible to obtain a positive semidefinite estimator
with a lower MSE by allowing some of the coefficients to be negative. Using
semidefinite programming (SDP), the nonnegativity constraint can be replaced
with a positive semidefiniteness constraint as
in~\cite{duFullyAutomaticComputation2010} as follows. The objective function
in~\eqref{eq:AQP} can be rewritten as $(\a - \B^{-1} \cv_k)^\top \B (\a -
\B^{-1} \cv_k) + \text{const.}$, where $\B = \D + \C$. By introducing an
auxiliary variable $t$ and constraint $(\a - \B^{-1} \cv_k)^\top \B (\a -
\B^{-1} \cv_k) \leq t$, the problem is converted into a minimization of $t$.
Using the Schur complement, the problem is reformulated as a semidefinite
program (SDP) in the variables $\a$ and $t$:
\begin{equation}\label{eq:SDP}
    \begin{array}{ll}
	\text{minimize} & t \\
	\text{subject to}
	& \begin{pmatrix}
	    t & (\a - \B^{-1} \cv_k)^\top \\
	    \a - \B^{-1} \cv_k & \B^{-1}
	\end{pmatrix} \succeq \zero \\
	& \LINPOOL(\a) \succeq \zero.
    \end{array}
\end{equation}
This is a convex optimization problem, which can be solved in polynomial time
with software such as CVX, which is a package for specifying and solving convex
programs~\cite{grantCVXMatlabSoftware2014,
grantGraphImplementationsNonsmooth2008}. \blue{There are cases where the
SDP~\eqref{eq:SDP}, nonnegativity constrained QP~\eqref{eq:AQP}, and
unconstrained problem~\eqref{eq:unconstrained} all give different solutions. In
these cases only the SDP and QP solutions yield positive semidefinite
estimators. In theory (when using the true $\C$ and $\D$), the SDP solution will
have a lower MSE. However, with estimation error in $\hat \C$ and $\hat \D$,
this is necessarily not the case. The computational complexity of the SDP
problem is significantly greater than that of the QP problem. The QP can be
computed with 2--3 orders of magnitude faster (depending on the problem
dimension). Due to the above reasons, we recommend the QP approach.}

\subsection{Extension to complex-valued data}
Extending the results to complex-valued data requires using the complex-valued
definitions of the elliptical kurtosis and the scaled MSE of the SCM. For a
review on complex elliptical symmetric (CES) distributions, see
e.g.,~\cite{ollilaComplexEllipticallySymmetric2012}. A CES distributed
(absolutely continuous) random vector $\x \in \complex^p$ from the $k$th
population has a density function up to a constant of the form
\begin{equation*}
    |\M_k|^{-1} h_k((\x - \bmu_k)^\hop \M_k^{-1} (\x - \bmu_k)),
\end{equation*}
where $h_k: \real_{\geq 0} \to \real_{> 0}$ is the density generator, $\bmu =
\E[\x]$ is the mean vector and  $\M_k = \E[ (\x - \bmu)(\x - \bmu)^\hop]$
denotes the Hermitian positive definite covariance matrix of $\x$.  Above
$(\cdot)^\hop$ denotes the Hermitian (complex-conjugate) transpose. We denote
this case by $\x \sim \complex \mathcal E_p(\bmu_k, \M_k, h_k)$. The definitions
of the SCM~\eqref{eq:SCM} and the SSCM~\eqref{eq:SSCM} stay unchanged except
that in their definitions $(\cdot)^\top$ is replaced with $(\cdot)^\hop$.
Furthermore, the inner products between Hermitian matrices ($\A = \A^\hop$)
remain unchanged since $\tr(\A\A^\hop) = \tr(\A^2)$. Hence, the definitions of
the sphericity parameter $\gamma_k$ as well as the scale $\eta_k$ remain
unchanged. The kurtosis of a complex marginal variable $x_i$ of $\x = (x_i)
\in \complex^p$ is defined as
\begin{equation*}
    \text{kurt}(x_i) = \frac{\E[|x_i-\mu_i|^4]}{\E[|x_i-\mu_i|^2]^2} - 2,
\end{equation*}
where $\mu_i = \E[x_i]$. The elliptical kurtosis is then
\begin{equation*}
    \kappa = (1/2) \text{kurt}(x_i).
\end{equation*}
The elliptical kurtosis $\kappa_k$ of class $k$ is estimated using the average
sample kurtosis of the marginal variables
\begin{equation}\label{eq:kappahatcomplex}
    \hat \kappa_k = 
    \frac{1}{2p}\sum_{j=1}^p
    \frac{\frac{1}{n_k}\sum_{i=1}^{n_k} |x_{ij,k} - \bar x_{j,k}|^4}
    {\left(\frac{1}{n_k}\sum_{i=1}^{n_k} |x_{ij,k}- \bar x_{j,k}|^2\right)^2} - 1,
\end{equation}
where $\bar x_{j,k} =\frac{1}{n_k}\sum_{i=1}^{n_k} x_{ij,k}$. The theoretical
lower bound of the kurtosis in the complex-valued case is $\kappa^{\text{LB}} =
-1/(p+1)$~\cite{ollilaComplexEllipticallySymmetric2012}. In case
\eqref{eq:kappahatcomplex} is less than $\kappa^{\text{LB}}$, we set $\hat
\kappa_k = 0.99 \kappa^{\text{LB}}$. 

The matrix $\D$ is estimated via~\eqref{eq:deltacomplex} given in the next
lemma, which is an extension of \cite[Lemma
1]{ollilaOptimalShrinkageCovariance2019} to the complex case.
\begin{lemma}
    \cite[Theorem 3]{raninenVariabilitySampleCovariance2021}
    Let $\mathcal X_k = \{\x_1,\ldots, \x_n\} \subset \complex^p$ be an i.i.d.
    random sample from $\mathcal \complex \mathcal E_p(\bmu_k,\M_k,h_k)$ with
    finite fourth-order moments. Then, the scaled MSE of the SCM is
    \begin{equation}
        \delta_k = \eta_k^2 \Big(
        \Big(\frac{1}{n_k - 1} + \frac{\kappa_k}{n_k}\Big) p
        + \frac{\kappa_k}{n_k} \gamma_k \Big).
        \label{eq:deltacomplex}
    \end{equation}
\end{lemma}

\section{Multi-target shrinkage estimators}\label{sec:multitarget}

Multi-target shrinkage covariance matrix estimators are capable of
simultaneously regularizing towards several target matrices. They are mainly
designed for the single population setting. However, some of the multi-target
shrinkage estimators can also be used in a multiclass setting for pooling SCMs.
Therefore, we give a brief review of the existing methods. Lastly, we show how
our proposed method can be applied to a single class multi-target shrinkage
covariance matrix estimation problem.

\subsection{Overview of multi-target shrinkage estimators}
There exist multi-target shrinkage covariance matrix estimators, which can be
used with user-defined target matrices and therefore also for pooling SCMs. In
this section, we discuss these estimators and in Section~\ref{sec:simulations}
we compare their performance to our proposed method by simulations.

Since the multi-target shrinkage estimators are developed for the single class
covariance matrix estimation setting, we define the data set as $\X =
\{\x_1,\ldots,\x_n\}$, where $(\x_i)_j = x_{ij}$. Let $\M$ denote the covariance
matrix and let $\S$ denote the SCM computed from $\X$. Multi-target shrinkage
covariance matrix estimators are often defined by
\begin{equation}\label{eq:MTS}
    \hat \M(\a) = a_0 \S + \sum_{k=1}^K a_k \T_k,
\end{equation}
where $\T_k$, $k=1,\ldots,K$, are \blue{linearly independent} target matrices
and $a_j$, $j=0,\ldots,K$, are the regularization coefficients.
In~\cite{lancewickiMultitargetShrinkageEstimation2014} and
\cite{bartzMultitargetShrinkage2014}, convex multi-target shrinkage covariance
matrix estimators were proposed, where $a_k \geq 0$ and $\sum_{k=1}^K a_k \leq
1$ for $k=1,\ldots,K$ and $a_0 = 1-\sum_{k=1}^K a_k$. The coefficients $\a =
(a_k)_{k=1}^K$ were chosen as the minimizers of the MSE loss function
\begin{equation*}
    L(\a; \Q, \mathbf b) = \E[\|\hat \M(\a) - \M \|^2_\Fro] = \a^\top
    \Q \a - 2\a^\top \mathbf b + \E[\|\S - \M\|^2_\Fro]
\end{equation*}
resulting in the constrained QP problem
\begin{equation*}
    \a^\star = \argmin_{\a \geq \zero, \one^\top \a \leq 1}
    L(\a; \Q, \mathbf b),~\text{where}
\end{equation*}
\begin{align*}
    (\Q)_{ij} = q_{ij} &= \E[\tr\big((\T_i-\S)(\T_j-\S)\big)]
\end{align*}
and $(\mathbf b)_i = b_i = \E[\tr\big((\T_i - \S)(\M - \S)\big)] = \E[\|\S -
\M\|_\Fro^2] - \E[\tr\big((\S - \M)(\T_i - \E[\T_i])\big)]$. As the elements in
$\Q$ and $\mathbf b$ depend on the unknown covariance matrix $\M$, they have to
be estimated. In~\cite{bartzMultitargetShrinkage2014}, the proposed estimates
were $\hat q_{ij} = \tr\big((\T_i-\S)(\T_j-\S)\big)$ and
\begin{equation*}
    \hat b_i \equiv \hat b
    = \sum_{i,j} \frac{1}{n(n-1)} \sum_{s=1}^n
    \Big( x_{si} x_{sj} - \frac{1}{n}\sum_{t=1}^n x_{ti} x_{tj} \Big)^2,
\end{equation*}
where the latter was obtained by approximating $\E[\tr\big((\S - \M)(\T_i -
\E[\T_i])\big)] \approx 0$. We compare our method with this method in the
simulations of Section~\ref{sec:simulations}, where it is denoted by BARTZ.
Regarding~\cite{lancewickiMultitargetShrinkageEstimation2014}, a specific
structural condition \cite[(10) and
(21)]{lancewickiMultitargetShrinkageEstimation2014} was imposed on the target
matrices, which prevents using it for pooling SCMs.

A leave-one-out-cross-validation (LOOCV) approach was considered
in~\cite{tongLinearShrinkageEstimation2018} for different scenarios. Their
proposition for SCM based linear shrinkage estimation used the LOOCV loss
function
\begin{align*}
    L_{\text{CV}}(\a; \Q_{\text{CV}}, \mathbf b_{\text{CV}})
    &=
    \frac{1}{n}\sum_{i=1}^n
    \big\| a_0 \S_{-i} + \sum_{k=1}^K a_k \T_k - \x_i \x_i^\top \big\|_\Fro^2
    \\
    &=
    \a^\top \Q_{\text{CV}} \a - 2\a^\top \mathbf b_{\text{CV}} + \text{const.},
\end{align*}
where $\mathbf{x}_i$ is assumed to have zero mean, $\mathbf{S}_{-i} =
\frac{1}{n} \sum_{j, j\neq i} \mathbf{x}_j\mathbf{x}_j^\top$ is the SCM computed
without the $i$th sample, and $\a = (a_0,a_1,\ldots,a_K)^\top \geq \zero$.
They showed how the elements in
$\Q_{\text{CV}}$ and $\mathbf b_{\text{CV}}$ can be computed analytically. They
also proposed a corresponding convex SCM based shrinkage estimator, which
requires that the targets have the same trace as the estimated SCM. In the
simulations of Section~\ref{sec:simulations}, we denote the linear estimator by
LOOCV.

Multi-target shrinkage estimators have also been proposed, for example
in~\cite{grayShrinkageEstimationLarge2018} from a Bayesian perspective and
in~\cite{zhangImprovedCovarianceMatrix2019} by regularizing the Gaussian
likelihood function. However, in both methods the target matrices are assumed to
be positive definite, and hence, SCMs cannot be used as targets when $n_k < p$.
In~\cite{riedlMultimodelShrinkageKnowledgeaided2018} a multi-target shrinkage
estimator of the form~\eqref{eq:MTS} was proposed for space-time adaptive
processing (STAP) problems, where the reliable estimation of the loss function
required additional prior information.

\subsection{Multi-target shrinkage estimation via pooling SCMs}\label{sec:arbitrarytargets}

As discussed above, some of the multi-target shrinkage estimators can be used in
order to linearly pool SCMs in a multi-class setting. Conversely, the proposed
method of linearly pooling SCMs can be used in single class covariance matrix
estimation as a multi-target shrinkage estimator as explained below.

Consider that there is only a single data set $\mathcal X$, its corresponding
SCM $\S$ and multiple positive definite symmetric target matrices
$\{\T_m\}_{m=1}^M$. Our approach for multi-target shrinkage estimation is
detailed in Algorithm~\ref{alg:linpool2}. The idea is to generate artificial
data sets $\{\mathcal Y_m\}_{m=1}^M$ with covariance matrices $\{\T_m\}_{m=1}^M$
so that $\X$ and the generated data sets $\{\mathcal Y_m\}_{m=1}^M$ are
approximately mutually independent. Specifically, $\mathcal Y_m$ is
conditionally independent of $\mathcal X$ given $\T_m$. Then, for the SCMs
$\{\S_{T_m}\}_{m=1}^M$ computed from $\{\mathcal Y_m\}_{m=1}^M$, we make the
approximations $\E[\tr(\S_{T_i} \S_{T_j})] \approx \tr(\E[\S_{T_i}]\E[\S_{T_j}])
= \tr(\T_i \T_j)$, for $i \neq j$, as well as $\E[\tr(\S \S_{T_j})] \approx
\tr(\E[\S]\E[\S_{T_j}]) = \tr(\M \T_j)$. We use the zero mean MVN distribution
to generate the data sets, i.e., $\mathcal Y_m \sim \mathcal N_p(\zero,\T_m)$,
$m=1,\ldots,M$, each consisting of $L$ samples.  We can then apply the proposed
method for pooling $\S$ and $\{\S_{T_m}\}_{m=1}^M$. We illustrate the usefulness
of this method in a portfolio optimization problem in
Section~\ref{sec:portfolio}, where it compares well with the other multi-target
methods.

\begin{algorithm}[!h]
    \caption{\text{Multi-target shrinkage estimator}}\label{alg:linpool2}
    \SetKwInOut{Input}{input}\SetKwInOut{Output}{output}
    \BlankLine
    \Input{Data $\X$, targets $\{\T_m\}_{m=1}^M$, sample size $L$, and
    lowerbound $\aLB>0$ for identity target.}
    \Output{$\hat \M$.}
    \BlankLine
    Generate i.i.d. samples $\mathcal Y_m \sim \mathcal N_p(\zero ,\T_m)$
    for $m=1,\ldots,M$ each of size $L$.

    Compute SCM $\S$ from $\mathcal X$ and SCMs $\S_{T_1}, \ldots, \S_{T_M}$
    from $\mathcal Y_1,\ldots, \mathcal Y_M$.

    Compute $\tilde\C$ and $\tilde\D$ of~\eqref{eq:tildeCD} (estimate as in Section~\ref{sec:estimation}).

    $\hat \a^\star \gets
        \underset{\a \geq \zero, a_{I} \geq \aLB}{\arg \min}
        \frac{1}{2} \a^\top (\tilde \D + \tilde \C) \a - \tilde \cv_k^\top \a$.

    $\hat \M \gets \LINPOOLI(\hat \a^\star)$ of \eqref{eq:LINPOOL}.
\end{algorithm}

\section{Simulation study}
\label{sec:simulations}

\setlength\tabcolsep{5.0pt}
\begin{table}
    \centering
    \caption{The normalized mean squared error over 1000 repetitions and
    standard deviation in the parenthesis.}\label{table:NMSE}
    \scriptsize
    \begin{tabular}{l l l l l l}
	\toprule
	& class 1 & class 2 & class 3 & class 4 & total\\
	\midrule
    \multicolumn{6}{l}{\underline{\emph{AR(1)}}}\vspace{0.2em}\\
	\input{results/AR1-NMSE-table.tex}\\
    \midrule
    \multicolumn{6}{l}{\underline{\emph{CS}}}\vspace{0.2em}\\
	\input{results/CS-NMSE-table.tex}\\
    \midrule
    \multicolumn{6}{l}{\underline{\emph{Mixed: $\M_1$ and $\M_2$ are AR(1);
        $\M_3$ and $\M_4$ are CS}}}\vspace{0.2em}\\
    \input{results/MIXED-NMSE-table.tex}\\
    \bottomrule
    \end{tabular}
\end{table}

This section provides several simulation studies in order to assess the MSE
performance of the proposed estimator as well as the accuracy of the plugin
estimates of $\D$ and $\C$ and the estimated coefficients $\hat\A$. We denote
the proposed linear pooling estimator~\eqref{eq:LINPOOL}, which includes
shrinkage towards identity, by LINPOOL. The LINPOOL estimator with an additional
convexity constraint on the coefficients ($\one^\top \a = 1$) is denoted by
LINPOOL-C. For LINPOOL and LINPOOL-C, we set the lowerbound for the identity
target in all of the simulations to $\aLBk = 10^{-8}$. In addition to the
proposed methods, the results are reported for the multi-target shrinkage method
LOOCV from~\cite{tongLinearShrinkageEstimation2018}, which uses a nonnegative
linear combination of the SCMs and the identity matrix, and BARTZ
from~\cite{bartzMultitargetShrinkage2014}, which uses a convex combination of
the SCMs and the identity matrix (see Section~\ref{sec:multitarget} for
details).
\pgfplotsset{deltaStyle/.style={
    boxplot/draw direction=x,
    ytick={1,2,3,4},
    yticklabels={
	\makebox[2.0ex][r]{$\hat\delta_1$},
	\makebox[2.0ex][r]{$\hat\delta_2$},
	\makebox[2.0ex][r]{$\hat\delta_3$},
	\makebox[2.0ex][r]{$\hat\delta_4$}},
    x tick label style={font=\small},
    y tick label style={font=\small},
    mark size=2,
    cycle list name={black white},
}}

\pgfplotsset{cijStyle/.style={
    boxplot/draw direction=x,
    ytick={1,2,3,4,5,6,7,8,9,10},
    yticklabels={
	\makebox[2.0ex][r]{$\hat c_{11}$},
	\makebox[2.0ex][r]{$\hat c_{12}$},
	\makebox[2.0ex][r]{$\hat c_{13}$},
	\makebox[2.0ex][r]{$\hat c_{14}$},
	\makebox[2.0ex][r]{$\hat c_{22}$},
	\makebox[2.0ex][r]{$\hat c_{23}$},
	\makebox[2.0ex][r]{$\hat c_{24}$},
	\makebox[2.0ex][r]{$\hat c_{33}$},
	\makebox[2.0ex][r]{$\hat c_{34}$},
	\makebox[2.0ex][r]{$\hat c_{44}$}},
    x tick label style={font=\small},
    y tick label style={font=\small},
    mark size=2,
    cycle list name={black white},
}}

\pgfplotsset{aijStyle/.style={
    boxplot/draw direction=x,
    ytick={1,2,3,4,5},
    yticklabels={
	\makebox[2.0ex][r]{$\hat a_{14}$},
	\makebox[2.0ex][r]{$\hat a_{24}$},
	\makebox[2.0ex][r]{$\hat a_{34}$},
	\makebox[2.0ex][r]{$\hat a_{44}$},
	\makebox[2.0ex][r]{$\hat a_{I4}$}},
    x tick label style={font=\small},
    y tick label style={font=\small},
    mark size=2,
    cycle list name={black white},
}}

\begin{figure}[t]
    \centering
    \begin{tikzpicture}
	\begin{groupplot}[
		name=DELTA,
		group style = {
		    group size = 2 by 3,
		    horizontal sep=30pt,
		    vertical sep=40pt,
		    },
		    width=4.8cm,
		    height=4cm,
		    ]
		    \nextgroupplot[deltaStyle, title={AR(1)},
		    title style = {yshift=3pt},
		    ]
		    \input{results/AR1-delta-boxplot.tex}
		    \nextgroupplot[deltaStyle, title={CS},
		    title style = {yshift=3pt},
		    ]
		    \input{results/CS-delta-boxplot.tex}
		    \nextgroupplot[cijStyle]
		    \input{results/AR1-cij-boxplot.tex}
		    \nextgroupplot[cijStyle]
		    \input{results/CS-cij-boxplot.tex}
		    \nextgroupplot[aijStyle]
		    \input{results/AR1-LIN2-a4-boxplot.tex}
		    \nextgroupplot[aijStyle]
		    \input{results/CS-LIN2-a4-boxplot.tex}
	\end{groupplot}
	\node[anchor=north] at ($(group c1r1.south east)!0.5!(group
	c2r1.south west)-(0,15pt)$) {{\small (a) estimates of $\delta_k$}};
	\node[anchor=north] at ($(group c1r2.south east)!0.5!(group
	c2r2.south west)-(0,15pt)$) {{\small (b) estimates of $c_{ij}$}};
	\node[anchor=north] at ($(group c1r3.south east)!0.5!(group
	c2r3.south west)-(0,15pt)$) {{\small (c) estimates of $a_{i4}$}};
    \end{tikzpicture}
    \caption{Estimates of $\delta_k$, $c_{ij}$, and $a_{i4}$ (coefficients for
    $\hat \M_4$ of LINPOOL). \emph{Left}: AR(1) setup. \emph{Right}: CS setup.
    The red triangles denote the theoretical values.}\label{fig:estimates}
\end{figure}
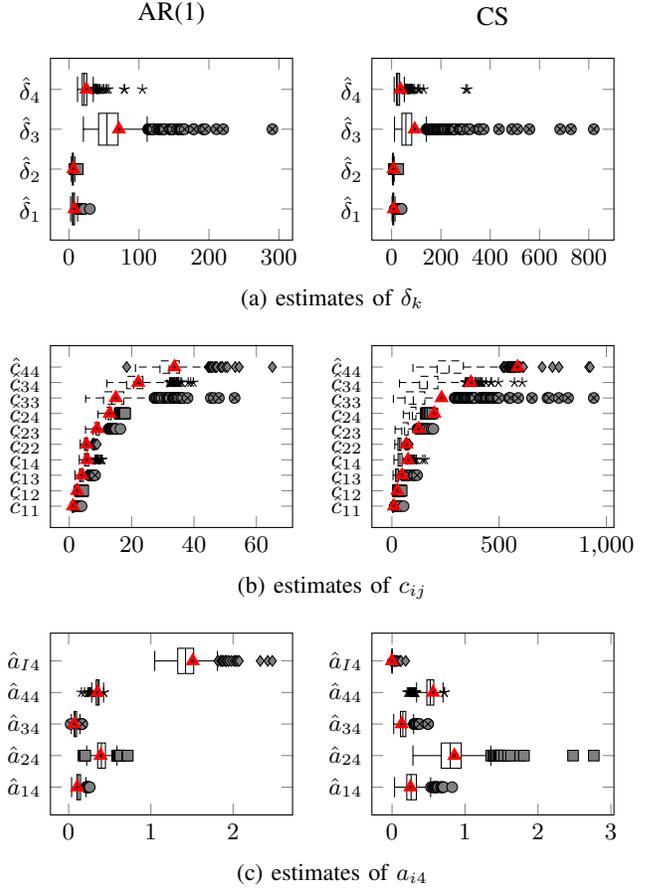

\pgfplotsset{totalnmsestyle/.style={
    width=4.8cm,
    height=4cm,
    xlabel near ticks,
    ylabel near ticks,
    tick label style={font=\footnotesize},
    label style={font=\footnotesize},
    legend style = {draw=none},
    scaled y ticks = true,
    yticklabel style={
        /pgf/number format/precision=3,
        /pgf/number format/fixed},
    every axis y label/.style={
        at={(rel axis cs:-0.25,0.5)},	% ylabel position
        rotate=90,			            % rotate label 90 degrees
        font=\footnotesize},		    % font size
    cycle list name={black white},
}}
\pgfplotstableread{results/totalNMSEcomplexAR1.dat}\totalNMSE
\pgfplotstableread{results/totalNMSEcomplexclassesAR1.dat}\totalNMSEclasses
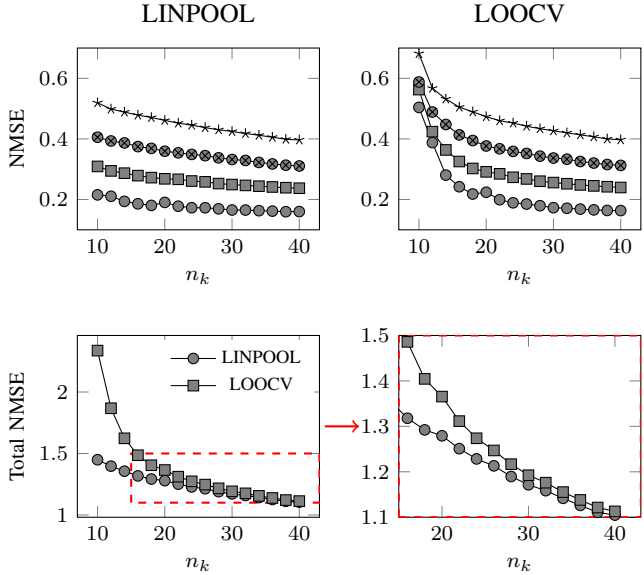
\begin{figure}[t]
    \centering
    \begin{tikzpicture}[scale=1]
        \begin{groupplot}[%
            group style = {group size = 2 by 2,
                           horizontal sep = 30pt,
                           vertical sep = 40pt}
                           ]
            \nextgroupplot[%
                totalnmsestyle,
                title = {LINPOOL},
                ymin = 0.1,
                ymax = 0.7,
                xlabel = {$n_k$},
                ylabel = {NMSE}]
            \addplot table [x=nvals, y=ave_SE_LINPOOL_1] {\totalNMSEclasses};
            \addplot table [x=nvals, y=ave_SE_LINPOOL_2] {\totalNMSEclasses};
            \addplot table [x=nvals, y=ave_SE_LINPOOL_3] {\totalNMSEclasses};
            \addplot table [x=nvals, y=ave_SE_LINPOOL_4] {\totalNMSEclasses};

            \nextgroupplot[%
                totalnmsestyle,
                title = {LOOCV},
                legend to name={totalnmseleg},
                legend style={draw=none,font=\scriptsize, legend columns = 3},
                ymin = 0.1,
                ymax = 0.7,
                xlabel = {$n_k$}]
            \addplot table [x=nvals, y=ave_SE_LOOCV_1] {\totalNMSEclasses};
            \addplot table [x=nvals, y=ave_SE_LOOCV_2] {\totalNMSEclasses};
            \addplot table [x=nvals, y=ave_SE_LOOCV_3] {\totalNMSEclasses};
            \addplot table [x=nvals, y=ave_SE_LOOCV_4] {\totalNMSEclasses};

            \nextgroupplot[%
                totalnmsestyle,
                xlabel = {$n_k$},
                ylabel = {Total NMSE}]
                legend style={draw=none,font=\scriptsize},
            \addplot table [x=nvals,
            y=total_SE_LINPOOL]{\totalNMSE};\addlegendentry{\scriptsize
            LINPOOL};
            \addplot table [x=nvals, y=total_SE_LOOCV]{\totalNMSE};
            \addlegendentry{\scriptsize LOOCV};
            \draw [red,dashed,thick] (15,1.1) rectangle (43,1.5);

            \nextgroupplot[%
                xmin=15,
                xmax=43,
                ymin=1.1,
                ymax=1.5,
                totalnmsestyle,
                xlabel = {$n_k$},
                ]
            \addplot table [x=nvals, y=total_SE_LINPOOL]{\totalNMSE};
            \addplot table [x=nvals, y=total_SE_LOOCV]{\totalNMSE};
            \draw [red,dashed,thick] (15,1.1) rectangle (43,1.5);
        \end{groupplot}
        \draw [thick,red,->,shorten >=15pt,shorten <=2pt] (group c1r2.east) --
        (group c2r2.west);
        \coordinate (c1) at ($(group c1r2.south west)!0.5!(group c2r2.south
        east)$);
        \node[below] at ($(c1)-(0,20pt)$)
        {\pgfplotslegendfromname{totalnmseleg}};
    \end{tikzpicture}
    \caption{NMSE as a function of sample size in the complex-valued AR(1) case
    with 4 classes. \emph{Top}: NMSE of individual classes. \emph{Bottom}: total
    combined NMSE.}\label{fig:totalNMSE}
\end{figure}

\subsection{Three different setups}
In the first simulation, we considered two different covariance matrix
structures: the AR(1), where $(\M_k)_{ij} = \sigma_k^2 \varrho_k^{|i-j|}$; and
the CS structure, where $(\M_k)_{ij} = \sigma_k^2 \varrho_k$, for $i\neq j$ and
$(\M_k)_{ij} = \sigma_k^2$ for $i = j$. We generated four $p=100$ dimensional
random samples of sizes $(n_1 = 20, n_2 = 100, n_3 = 20, n_4 = 100)$ from four
independent multivariate $t$-distributions with $\nu=8$ degrees of freedom. The
means of the classes were generated from the standard MVN distribution and held
constant over the repetitions. We simulated three different setups. In the first
setup, all covariance matrices had an AR(1) structure. In the second setup, all
covariance matrices had an CS structure. In the last mixed setup, classes $1$
and $2$ had an AR(1) structure and classes $3$ and $4$ had a CS structure.
For all setups, we used $\sigma_k^2 = k$ and $(\varrho_1 = 0.3, \varrho_2 = 0.4,
\varrho_3 = 0.5, \varrho_4 = 0.6)$.

Table~\ref{table:NMSE} tabulates the normalized mean squared error,
\begin{equation*}
    \NMSE(\hat\M_k) =
    \text{Ave}\|\hat \M_k - \M_k \|^2_\Fro / \|\M_k\|^2_\Fro,
\end{equation*}
and total NMSE (sum of NMSEs of the classes) of the covariance matrix estimates
for the three different setups over 1000 repetitions. Table~\ref{table:NMSE}
shows that the proposed method LINPOOL performed best in the AR(1) setup and the
mixed setup, whereas LINPOOL-C performed best (slightly better than LINPOOL) in
the CS setup.

Figure~\ref{fig:estimates} displays boxplots of the estimates $\hat\delta_k$ and
$\hat c_{ij}$ both for the AR(1) case as well as the CS case. The estimates of
the optimal coefficients for LINPOOL for the fourth class $\hat a_{i4}$ are also
shown. As can be seen from the boxplots, for the AR(1) case, the medians of the
estimates were mostly correctly placed over the true values, which are denoted
by the red triangles ({\color{red} $\blacktriangle$}). For the CS case, there
was some significant bias in the estimation of $c_{ij}$, which is due to the
fact that the assumption \eqref{eq:gamma_oh_p} does not hold in this case as
explained in Section~\ref{sec:signestimators}. Despite of this, the final
coefficients $\hat a_{ij}$ were reasonable well estimated.  

Here, it is good to note that, for the CS case, the coefficient
corresponding to identity shrinkage $\hat a_{I4}$ is close to zero. When this
happens, there is a possibility that (despite having a low MSE) the estimate is
not well-conditioned resulting in high error when inverting the estimate.
Therefore, in these cases (depending on the conditioning of the estimate) it can
be useful to increase the lower bound $\aLB_4$. Generally, for class $k$, one
could use $\aLBk = \alpha \eta_k$, where $\alpha \in [0,1]$.

\blue{
\subsection{Increasing the number of classes}

Next, we examine the NMSE of the LINPOOL estimator of class 1 as the number of
classes $K$ increase from 2 to 16 classes. The setup is as follows. The first
class has an AR(1) covariance matrix structure with a fixed parameter
$\varrho_1=0.5$. The other class covariance matrices also have an AR(1)
structure except for the classes $k=4,8,12,16$, which have a CS structure. The
parameter $\varrho_k$, for $k\geq 2$, is chosen uniformly at random from the
interval $[0.1, 0.6]$ for each Monte Carlo trial. The means of the classes were
generated from the standard MVN distribution for each Monte Carlo trial. The
sample sizes are equal with $n_k=40$ for all classes $k$. The dimension is
$p=100$ and the data is multivariate $t$-distributed with $\nu=8$ degrees of
freedom. 

Figure~\ref{fig:varK} depicts the results averaged over 1000 Monte Carlo trials
for each value of $K$. The red vertical lines mark the spots when the added
class covariance matrix has an CS structure (i.e., doesn't share the same
structure as class 1). One can observe that every time an AR(1) structured class
is added, the NMSE decreases. When the added class has a different covariance
structure (the CS structure), the NMSE does not decrease. An exception to this
is when $K=4$. A reason for this may be that the total number of observations is
still relatively low and including the fourth class helps in reducing the
variance of the estimate.
}

\pgfplotstableread{results/incK.dat}\incKtable
\begin{figure}
    \centering
    \begin{tikzpicture}
        \begin{axis}[xlabel={number of classes $K$},
            ylabel={$\NMSE$ of $\hat\M_1$},
            width = 8.858333cm,
            height = 4cm,
            xlabel near ticks,
            ylabel near ticks,
            tick label style={font=\footnotesize},
            label style={font=\footnotesize},
            legend style = {draw=none},
            scaled y ticks = true,
            yticklabel style={
                /pgf/number format/precision=3,
                /pgf/number format/fixed},
            cycle list name={black white},
            ]
            \addplot table[x=Karr, y=NMSE1] {\incKtable};
            \draw[red] (axis cs:4,0) -- (axis cs:4,0.34);
            \draw[red] (axis cs:8,0) -- (axis cs:8,0.34);
            \draw[red] (axis cs:12,0) -- (axis cs:12,0.34);
            \draw[red] (axis cs:16,0) -- (axis cs:16,0.34);
        \end{axis}
    \end{tikzpicture}
    \caption{NMSE of $\hat\M_1$ as the number of classes increase.}\label{fig:varK}
\end{figure}
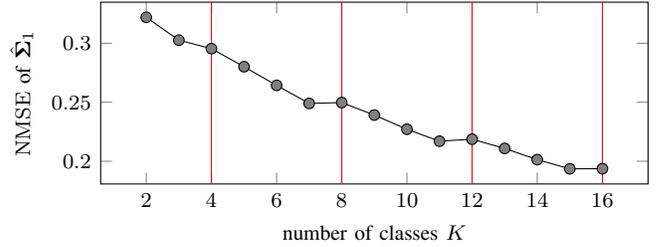

\subsection{Complex-valued case}
In the next simulation, the classes have an AR(1) covariance matrix structure,
\blue{$(\M_k)_{ij} = \sigma_k^2 \varrho_k^{|i-j|}$ for $i \leq j$, and
$(\M_k)_{ji} = (\M_k^*)_{ij}$, for $i>j$, where $(\cdot)^*$ denotes complex
conjugation.} The used parameters were $\sigma_k^2 =
k$, $\varrho_1 = 0.3 e^{\jmath 2\pi 0.3}$, $\varrho_2 = 0.4 e^{\jmath 2\pi
0.4}$, $\varrho_3 = 0.5 e^{\jmath 2\pi 0.5}$, and $\varrho_4 = 0.6 e^{\jmath
2\pi 0.6}$. We simulated the NMSE as a function of the sample size from $n_k \in
\{10, 15, \ldots 40\}$ for all $k$. The data was generated from the complex
multivariate $t$-distribution with $\nu=8$ degrees of freedom and dimension
$p=100$. The results were averaged over 1000 Monte Carlo trials for each sample
size and are shown in Figure~\ref{fig:totalNMSE}. It can be observed that
especially for small sample sizes ($n_k < 30$) LINPOOL performed better than
LOOCV.

\section{Portfolio optimization} \label{sec:portfolio}

\pgfplotsset{portfoliostyle/.style={
        width=4.8cm,
        height=4cm,
        xlabel near ticks,
        ylabel near ticks,
        tick label style={font=\footnotesize},
        label style={font=\footnotesize},
        scaled y ticks = true,
        yticklabel style={
            /pgf/number format/precision=3,
        /pgf/number format/fixed},
        every axis y label/.style={
            at={(rel axis cs:-0.25,0.5)},	% ylabel position
            rotate=90,			            % rotate label 90 degrees
        font=\footnotesize},		    % font size
        cycle list name={black white},
}}
\newcommand{\portplots}{%
    \addplot table [x=n,y=LW-well]{\loadedtable};
    \addplot table [x=n,y=LW-improved]{\loadedtable};
    \addplot table [x=n,y=LW-honey]{\loadedtable};
    \addplot[teal,mark=triangle*] table [x=n,y=LW-analytical]{\loadedtable};
    \addplot[blue,mark=*] table [x=n,y=LOOCV]{\loadedtable};
    \addplot[green,mark=square*] table [x=n,y=BARTZ]{\loadedtable};
    \addplot[red,mark=oplus*] table [x=n,y=LINPOOL]{\loadedtable};
    \addplot[red,mark=triangle*] table [x=n,y=LINPOOLC]{\loadedtable};
}

\newcommand{\portplotswleg}{
    \addplot table [x=n,y=LW-well]{\loadedtable};\addlegendentry{LW-well}
    \addplot table [x=n,y=LW-improved]{\loadedtable};\addlegendentry{LW-improved}
    \addplot table [x=n,y=LW-honey]{\loadedtable};\addlegendentry{LW-honey}
    \addplot[teal,mark=triangle*] table [x=n,y=LW-analytical]{\loadedtable};\addlegendentry{LW-analytical};
    \addplot[blue,mark=*] table [x=n,y=LOOCV]{\loadedtable};\addlegendentry{LOOCV}
    \addplot[green,mark=square*] table [x=n,y=BARTZ]{\loadedtable};\addlegendentry{BARTZ};
    \addplot[red,mark=oplus*] table [x=n,y=LINPOOL]{\loadedtable};\addlegendentry{LINPOOL}
    \addplot[red,mark=triangle*] table [x=n,y=LINPOOLC]{\loadedtable};\addlegendentry{LINPOOL-C}
}
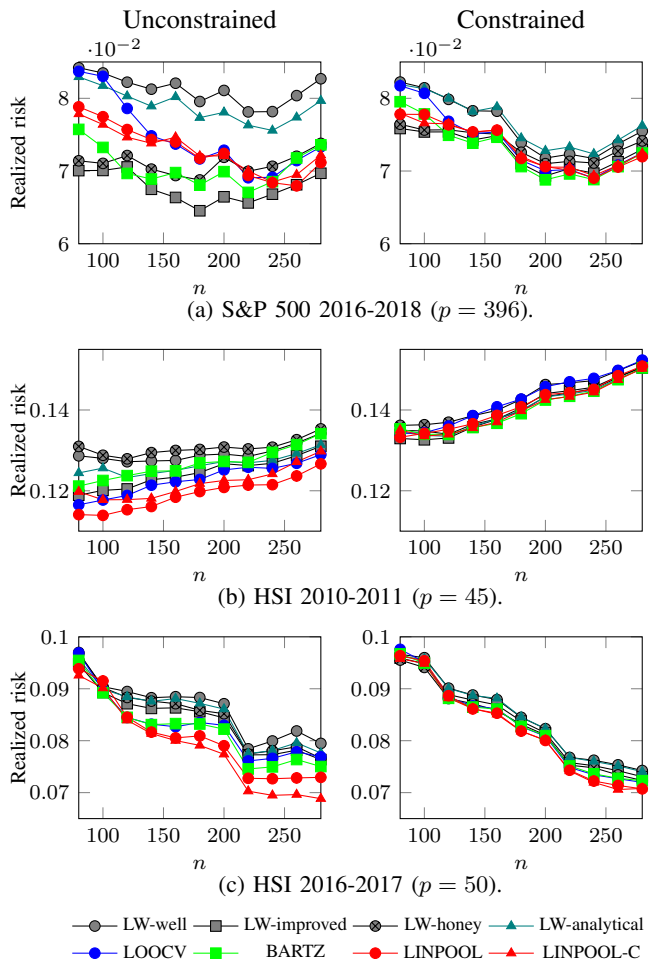
\begin{figure}[t]
    \centering
    \begin{tikzpicture}[scale=1]
        \begin{groupplot}[group style = {
                group size = 2 by 3,
                horizontal sep = 30pt,
            vertical sep = 40pt},
            ]
            \nextgroupplot[title={Unconstrained},
            title style = {yshift=3pt},
            portfoliostyle,
            xlabel = {$n$},
            ylabel = {Realized risk},
            xmin = 80,
            xmax = 280,
            ymin = 0.06,
            ymax = 0.085,
            ]
            \pgfplotstableread{results/minvarportfolioresultsSP500noconstraintsRISK.dat}\loadedtable
            \portplots
            \nextgroupplot[title={Constrained},
            title style = {yshift=3pt},
            portfoliostyle,
            xlabel = {$n$},
            ylabel = {\mbox{}},
            xmin = 80,
            xmax = 280,
            ymin = 0.06,
            ymax = 0.085,
            ]
            \pgfplotstableread{results/minvarportfolioresultsSP500longonlywithconstraintsRISK.dat}\loadedtable
            \portplots
            \nextgroupplot[portfoliostyle,
            xlabel = {$n$},
            ylabel = {Realized risk},
            xmin = 80,
            xmax = 280,
            ymin = 0.11,
            ymax = 0.155,
            ]
            \pgfplotstableread{results/minvarportfolioresultsHSI_y2010noconstraintsRISK.dat}\loadedtable
            \portplots
            \nextgroupplot[portfoliostyle,
            xlabel = {$n$},
            ylabel = {\mbox{}},
            xmin = 80,
            xmax = 280,
            ymin = 0.11,
            ymax = 0.155,
            ]
            \pgfplotstableread{results/minvarportfolioresultsHSI_y2010longonlywithconstraintsRISK.dat}\loadedtable
            \portplots
            \nextgroupplot[portfoliostyle,
            xlabel = {$n$},
            ylabel = {Realized risk},
            xmin = 80,
            xmax = 280,
            ymin = 0.065,
            ymax = 0.10,
            ]
            \pgfplotstableread{results/minvarportfolioresultsHSI_y2016noconstraintsRISK.dat}\loadedtable
            \portplots
            \nextgroupplot[portfoliostyle,
            xlabel = {$n$},
            ylabel = {\mbox{}},
            xmin = 80,
            xmax = 280,
            ymin = 0.065,
            ymax = 0.10,
            legend to name={portleg},
            legend style={draw=none,font=\scriptsize, legend columns = 4},
            ]
            \pgfplotstableread{results/minvarportfolioresultsHSI_y2016longonlywithconstraintsRISK.dat}\loadedtable
            \portplotswleg
        \end{groupplot}

        \coordinate (c1) at ($(group c1r3.south west)!0.5!(group
        c2r3.south east)$);
        \node[below] at ($(c1)-(0pt,30pt)$) {\pgfplotslegendfromname{portleg}};
        \node[anchor=north] at ($(group c1r1.south east)!0.5!(group
        c2r1.south west)-(0,18pt)$) {
        {\small (a) S\&P 500 2016-2018 ($p=396$).}};
        \node[anchor=north] at ($(group c1r2.south east)!0.5!(group
        c2r2.south west)-(0,18pt)$) {
        {\small (b) HSI 2010-2011 ($p=45$).}};
        \node[anchor=north] at ($(group c1r3.south east)!0.5!(group
        c2r3.south west)-(0,18pt)$) {
        {\small (c) HSI 2016-2017 ($p=50$).}};
    \end{tikzpicture}
    \caption{Annualized realized GMVP risk achieved out-of-sample for
        different covariance matrix estimators and different training window
        lengths $n$. \emph{Left:} unconstrained portfolio. \emph{Right}:
        constrained portfolio (nonnegative weights and maximum single asset weight
    0.1).}\label{fig:portfolio}
\end{figure}

We studied the performance of the proposed method in a portfolio optimization
problem using divident adjusted daily closing prices.  Portfolio optimization is
a central topic in investment theory, see,
e.g.,~\cite{markowitzPortfolioSelection1952,
    markowitzPortfolioSelectionEfficient1959, tobinLiquidityPreferenceBehavior1958,
sharpeCapitalAssetPrices1964}, and~\cite{lintnerValuationRiskAssets1965}. A
focus in portfolio optimization has been on the estimation of the covariance
matrix of the stock returns, commonly using shrinkage regularization techniques
or random matrix theory, see,
e.g.,~\cite{ledoitImprovedEstimationCovariance2003,ledoitWellconditionedEstimatorLargedimensional2004,
    ledoitAnalyticalNonlinearShrinkage2020, ledoitSpectrumEstimationUnified2015,
yangRobustStatisticsApproach2015},
and~\cite{fengSignalProcessingPerspective2016}. In a portfolio optimization
problem, a particular investment portfolio is determined by a weight or
allocation vector $\w \in \real^p$ (verifying the constraint $\one^\top \w = 1$)
whose elements describe the fraction of the total wealth invested in each of the
$p$ stocks. We considered two different portfolios. First, we considered the
\emph{global minimum variance portfolio} (GMVP) in which one seeks a portfolio
$\w$ that minimizes the risk (variance). The optimization problem is thus
\begin{equation}\label{eq:portfolioweights}
    \underset{\w \in \real^p}{\text{minimize}} ~ \w^\top \M \w
    \quad
    \text{subject to} ~ \one^\top \w = 1.
\end{equation}
The well-known solution is $\w^\star = \frac{\M^{-1} \one}{\one^\top \M^{-1}
\one}$, where $\M$ is the covariance matrix of the stock returns. We also
considered a \emph{constrained} portfolio, where the coefficients are
constrained to be within the range $0 \leq w_i \leq 0.1$, for all $i$, i.e.,
shorting (negative weights) is not allowed and the portfolio manager is not
allowed to put more than 10\% of the wealth in one stock. The optimization
problem for this case is the same as in~\eqref{eq:portfolioweights} but having
the additional constraint $\zero \leq \w \leq 0.1 \cdot \one$, which results in
a QP problem. 

In the simulation, the covariance matrix $\M$ was estimated via a sliding
window method so that at day $t$ it was estimated using the daily net returns
of the previous $n$ days from $t-n$ to $t-1$. The portfolio weights were then
computed via~\eqref{eq:portfolioweights} with and without the additional
weight constraints. These yielded the portfolio returns for the next 20
(trading) days. Then, the sliding window shifted 20 days forward and the
procedure was repeated. By denoting the total number of days in the data by
$T$, we obtained $T-n$ daily returns from which we computed the realized risk
as the empirical standard deviation of the daily portfolio returns. 

We applied the proposed method (explained in Section~\ref{sec:arbitrarytargets}
and Algorithm~\ref{alg:linpool2}) for single class covariance matrix estimation
using the same target matrices for regularization as
in~\cite{ledoitImprovedEstimationCovariance2003}
and~\cite{ledoitHoneyShrunkSample2004}. The target matrices were the single
factor market index model $\T_F$
of~\cite{ledoitImprovedEstimationCovariance2003} and the constant correlation
model $\T_C$ of~\cite{ledoitHoneyShrunkSample2004}. Their computation is
explained in~\cite{ledoitImprovedEstimationCovariance2003}
and~\cite{ledoitHoneyShrunkSample2004}, respectively. At day $t$, we used the
$n$ previous days to compute the SCM. However, due to the nature of our method,
we were able to freely choose the amount of data used for computing the
regularization target matrices. Hence, we chose to use the previous 40 days
($t-40$ to $t-1$ corresponding to the previous two months) for computing $\T_F$
and $\T_C$, regardless of the window size $n$ ($n > 40$) used for estimating the
SCM. This can be justified by the fact that the trend of the market is better
captured by the most recent net returns. After computing $\T_F$ and $\T_C$, we
generated 1000 i.i.d. samples both from $\mathcal Y_C \sim \mathcal
N(\zero,\T_C)$ and $\mathcal Y_F \sim \mathcal N(\zero, \T_F)$ and estimated the
coefficients for the proposed LINPOOL estimator $\LINPOOLI(\a) = a_1 \S + a_2
\S_F + a_3 \S_C + a_I \I$, where $\S_F$ and $\S_C$ denote the SCMs computed from
$\mathcal Y_F$ and $\mathcal Y_C$, respectively. We also report the performance
of the LINPOOL estimator~\eqref{eq:LINPOOL} using a convex combination as
explained in Remark~\ref{remark:convexpooling}. For both methods, we used the
lowerbound constraint $a_{I} \geq \aLB = 10^{-8}$. We also report the
performance of the multi-target shrinkage estimation methods LOOCV
of~\cite{tongLinearShrinkageEstimation2018} and BARTZ
of~\cite{bartzMultitargetShrinkage2014} using the same target matrices ($\T_F$,
$\T_C$, and $\I$) as for our proposed method. Additionally, we report the
performance of the following methods specifically tailored for portfolio
optimization: LW-well
of~\cite{ledoitWellconditionedEstimatorLargedimensional2004}, LW-improved
of~\cite{ledoitImprovedEstimationCovariance2003}, LW-honey
of~\cite{ledoitHoneyShrunkSample2004}, and the random matrix theory based
estimator LW-analytical of~\cite{ledoitAnalyticalNonlinearShrinkage2020}.

We used the same portfolio data sets as
in~\cite{ollilaOptimalShrinkageCovariance2019} \blue{(obtained from
\url{https://finance.yahoo.com})}. That is, Standard \& Poor's 500 stock market
index (S\&P 500) tracking $396$ stocks from from Jan. 4, 2016 to Apr. 27, 2018
consisting of $T=583$ days. Hang Seng Index (HSI) from Jan. 4, 2010 to Dec. 24,
2011 ($45$ stocks during $T=491$ days) and from Jan. 1, 2016 to Dec. 27, 2017
($50$ stocks during $T= 489$ days after removing two zero return days). 

Figure~\ref{fig:portfolio} shows the annualized risk obtained by the different
estimators as a function of the training window length $n$.  Regarding the S\&P
500 2016-2018 data, there were large differences in the performances of the
methods.  LW-improved performed best and LW-well the worst.  The differences
between the methods were smaller in the constrained case, where BARTZ performed
best.  Regarding the HSI 2010-2011 data set, the proposed method LINPOOL
achieved the lowest risk for the unconstrained portfolio with all window sizes
$n$. For the constrained portfolio, all of the methods performed nearly equally
well with LW-improved having the lowest risk with the window length $n=100$.
Regarding the HSI 2016-2017 data, the proposed methods (LINPOOL and LINPOOL-C)
achieved the lowest risk for both the constrained and unconstrained portfolios
for all window sizes $n > 120$.

\section{Conclusions}\label{sec:conclusions}
The paper proposed a regularized sample covariance matrix estimation method for
high-dimensional multiclass problems. The proposed estimator was formed from a
linear combination of the class SCMs. We derived the theoretically optimal
coefficients that minimize the mean squared error. The optimal coefficients
depend on unknown parameters, and their estimation was addressed under the
assumption that the samples are generated from unknown elliptically symmetric
populations with finite fourth-order moments. In constructing estimators for the
unknown parameters, we utilized the sample spatial sign covariance matrix, which
we showed in Theorem~\ref{thm:sign} to be an asymptotically unbiased estimator
of the normalized covariance matrix in the case that the sphericity parameter of
the distribution grows slower than the dimension. The effectiveness and
usefulness of the proposed method was demonstrated via simulations and a
portfolio optimization problem using real stock data. Codes are available at
\href{https://github.com/EliasRaninen}{github.com/EliasRaninen}.

\appendices
\section{Proof of Theorem~\ref{thm:MSE}}\label{app:thm:MSE}
Write $f(\a)= \E[ \| \LINPOOL(\a) - \M_k \|_\Fro^2]$ for the objective function.
By expanding the expression for the squared error, we get
\begin{align*}
        &\| \LINPOOL(\a) - \M_k \|^2_\Fro
        = \tr\Big( \Big(\sum_{i=1}^K a_{i} \S_i - \M_k\Big)
        \Big( \sum_{j=1}^K a_{j} \S_j - \M_k\Big)^\top \Big)
        \\
        &= \sum_{i=1}^K \sum_{j=1}^K a_{i}a_{j} \tr(\S_i \S_j)
        - 2\sum_{j=1}^K a_{j} \tr(\S_j \M_k) + \tr(\M_k^2).
\end{align*}
Taking the expectation and scaling by $1/p$ gives
\begin{equation*}
    (1/p) f(\a) = \a^\top \B \a - 2 \cv_k^\top \a + c_{kk},
\end{equation*}
where
\begin{align*}
    \B &= (b_{ij}) \in \real^{K \times K}, ~ b_{ij} = \E[\tr(\S_i \S_j)]/p,
    \\
    \C &= (c_{ij}) \in \real^{K \times K}, ~ c_{ij} = \tr(\M_i \M_j)/p,
\end{align*}
and $\cv_k$ corresponds to the $k$th column of $\C$. For $i \neq j$,
\begin{align*}
    b_{ij} &= p^{-1} \E[ \tr(\S_i \S_j) ] = p^{-1} \tr(\E[\S_i]\E[\S_j]) \\
           &= p^{-1} \tr(\M_i\M_j) = c_{ij}.
\end{align*}
Using that $\E[\tr(\S_i^2)] = \mathrm{MSE}(\S_i) + \tr(\M^2_i)$, we get for
$i=j$,
\begin{equation*}
    b_{ii} = p^{-1} \E[\tr(\S_i^2)] = \delta_i + c_{ii},
\end{equation*}
and so $\B = \D + \C$. By definition $\D$ is symmetric and positive definite.
Also, $\C$ is symmetric and positive semidefinite. This follows since for any
$\m \in \real^K$, $\m^\top \C \m = \sum_{i=1}^K \sum_{j=1}^K m_i m_j p^{-1}
\tr(\M_i \M_j) = \tr(\M_*^2) \geq 0$, where $\M_* = p^{-1/2}\sum_{j=1}^K m_j
\M_j$. Hence $\D + \C \succ \zero$ is invertible.

Regarding the extension discussed in Section~\ref{sec:extensions}, we now show
the positive definiteness of $\tilde \D + \tilde \C$. Since we know that $\D +
\C \succ \zero$, due to the properties of the Schur
complement~\cite[A.5.5]{boydConvexOptimization2004}, it holds that
\begin{equation*}
    \tilde \D + \tilde \C =
    \begin{pmatrix}
        \D + \C & \f \\
        \f^\top & 1
    \end{pmatrix}
    \succ \zero
    \Leftrightarrow
    \D + \C - \f \f^\top \succ \zero.
\end{equation*}
We can then show the positive definiteness of $\tilde \D + \tilde \C$ by showing
that $\D + \C - \f \f^\top$ is positive definite. For $\m \in \real^K \setminus
\{\zero\}$,
\begin{align*}
        &\m^\top (\D + \C - \f \f^\top)\m
        =
        \sum_k m_k^2 \delta_k
        + \sum_{i,j} m_i m_j (c_{ij} - \eta_i\eta_j)
        \\
        &=
        \sum_{k} m_k^2 \delta_k + \tr (\tilde \M_*^2) > 0,
\end{align*}
where $\tilde \M_* = p^{-1/2} \sum_{j=1}^K m_j (\M_j - \eta_j \I)$.
So, $\tilde\D+\tilde\C \succ \zero$. 
\qed

\section{Proof of Theorem~\ref{thm:sign}}\label{app:thm:sign}
Let $\x \sim \mathcal E_p(\bmu, \M, g)$ and assume $\SSCMshape$ is computed
with a known mean $\bmu=\zero$ in~\eqref{eq:SSCM}. Then $\E[\SSCMshape] =
p\E[\x \x^\top /\|\x\|^2] \equiv \SSCMshapepop$. Let $\u = (u_i) =
\M^{-1/2}\x/\|\M^{-1/2}\x\|$. So, $\u$ is uniformly distributed on the unit
sphere $\{\u \in \real^p : \|\u\| = 1\}$. Denote $\shape = (\Lambda_{ij})$.  We
need the following lemma.
\begin{lemma}\label{lemma:result}
    \begin{equation*}
        \E\left[ \shape^{1/2} \u\u^\top \shape^{1/2} \u^\top \shape \u \right]
        = p^{-1}(p+2)^{-1} \left( \tr(\shape)\shape + 2\shape^2 \right).
    \end{equation*}
\end{lemma}
\begin{proof}
    \begin{align*}
            &\E\left[ \shape^{1/2}\u\u^\top\shape^{1/2} \u^\top \shape \u \right]
            =
            \sum_{k,l} \Lambda_{kl} \shape^{1/2} \E[ u_k u_l \u\u^\top] \shape^{1/2}
            \\
            &=
            \shape^{1/2} \Big(
                \sum_{l} \Lambda_{ll} \E[ u_l^2 \u\u^\top]
                +
                \sum_{k \neq l} \Lambda_{kl} \E[ u_k u_l \u\u^\top]
            \Big) \shape^{1/2}.
            \numberthis
            \label{eq:wholething}
    \end{align*}
    Here (see~\cite[3.1.2]{fangSymmetricMultivariateRelated2018}),
    \begin{equation*}
        (\E[ u_l^2 \u\u^\top])_{jj} = \begin{cases}
            \E[u_l^4] = 3 p^{-1}(p+2)^{-1}, & l = j, \\
            \E[u_l^2 u_j^2] = p^{-1}(p+2)^{-1}, & l \neq j, \\
        \end{cases}
    \end{equation*}
    and $(\E[ u_l^2 \u\u^\top])_{ij} = 0$, for $i\neq j$. Hence, $\E[ u_l^2
    \u\u^\top] = p^{-1}(p+2)^{-1}(\I + 2\e_l\e_l^\top)$. Thus,
    \begin{align*}
            &\sum_{l} \Lambda_{ll} \E[ u_l^2 \u\u^\top]
            \\
            &= p^{-1}(p+2)^{-1}( (\sum_{l} \Lambda_{ll}) \I + 2 \sum_{l}
            \Lambda_{ll} \e_l\e_l^\top )
            \\
            &= p^{-1}(p+2)^{-1}(\tr(\shape)\I + 2\diag(\shape)).
            \numberthis
            \label{eq:firstterm}
    \end{align*}
    Regarding the other term in~\eqref{eq:wholething}, where $k\neq l$,
    \begin{align*}
            &(\E[ u_k u_l \u\u^\top])_{mn}
            \\
            &= \begin{cases}
                \E[u_k^2 u_l^2] = p^{-1}(p+2)^{-1}, & k=m, l=n, \\
                \E[u_k^2 u_l^2] = p^{-1}(p+2)^{-1}, & k=n, l=m, \\
                0, & \text{otherwise}.
            \end{cases}
    \end{align*}
    Hence, we have
    \begin{equation}\label{eq:secondterm}
        \sum_{k\neq l} \Lambda_{kl} \E[ u_k u_l \u\u^\top]
        = 2 p^{-1}(p+2)^{-1} (\shape - \diag(\shape))
    \end{equation}
    and the result follows by substituting \eqref{eq:firstterm} and
    \eqref{eq:secondterm} into \eqref{eq:wholething}.
\end{proof}

Note that,~\eqref{eq:firstterm},~\eqref{eq:secondterm},
and~\eqref{eq:wholething} are valid more generally for any positive semidefinite
symmetric matrix, not only for the shape matrix. 

We are ready to prove Theorem~\ref{thm:sign}. 
Since for any $x>0$, $x^{-1} \geq 2 - x$, we have
\begin{align*}
    \SSCMshapepop 
        &= p \E\left[\frac{\x\x^\top}{\|\x\|^2} \right]
        = p\E\Bigg[\frac{\shape^{1/2}\u\u^\top\shape^{1/2}}{\u^\top
        \shape \u}\Bigg]
        \\
        &\succeq
        p\E\left[\shape^{1/2}\u\u^\top\shape^{1/2} (2 - \u^\top \shape \u) \right]
        \\
        &=
        2p\E[ \shape^{1/2}\u\u^\top\shape^{1/2} ]
        - p\E[\shape^{1/2}\u\u^\top\shape^{1/2}\u^\top \shape \u]
        \\
        &=
        2 \shape - \frac{1}{p+2}\left( \tr(\shape)\shape + 2\shape^2
        \right)
        =
        \shape - \frac{2(\shape^2 - \shape)}{p+2}.
        \numberthis
        \label{eq:approximation}
\end{align*}
By scaling both sides of~\eqref{eq:approximation} by $\epsilon =
\|\shape\|_\Fro^{-1} = 1/\sqrt{p\gamma}$, the first term will have unity norm
($\|\epsilon \shape\|_\Fro = 1$). Let us then consider the second term on the
right-hand side. Its trace is
\begin{align*}
    \epsilon \frac{2\tr(\shape^2 - \shape)}{p+2}
        &= \epsilon \frac{2p}{p+2}(\gamma-1) 
        \\
        &\leq 2\left(\sqrt{\frac{\gamma}{p}} - \sqrt{\frac{1}{p\gamma}}\right)
        = O \left (\sqrt{\frac{\gamma}{p}} \right),
        \numberthis
        \label{eq:tr2term}
\end{align*}
and its norm is
\begin{align*}
    \epsilon \left\| \frac{2(\shape^2 - \shape)}{p+2} \right\|_\Fro
        &\leq (2\epsilon/p) ( \|\shape^2\|_\Fro + \|\shape\|_\Fro )
        \\
        &\leq
        (2/p) (\|\shape\|_\Fro + 1)
        = O(\sqrt{\gamma/p}).
        \numberthis
        \label{eq:norm2term}
\end{align*}
Here, we used the triangle inequality and submultiplicativity properties of the
Frobenius norm. By moving all the terms of~\eqref{eq:approximation} to the
left-hand side and scaling them by $\epsilon$, we can use \eqref{eq:tr2term} and
\eqref{eq:norm2term} as well as the property that for any $\A \succeq \zero$,
$\tr(\A) \geq \|\A\|_\Fro$ to get
\begin{equation*}
    O(\sqrt{\gamma/p}) \geq
    \left\|
    \left(\SSCMshapepop - \shape\right)/\|\shape\|_\Fro
    + O(\sqrt{\gamma/p})
    \right\|_\Fro
    \geq
    0,
\end{equation*}
which implies that $\SSCMshapepop = \shape + O(\gamma)$. Furthermore, if
$\gamma=o(p)$, then $O(\sqrt{\gamma/p}) = o(1)$ implying $\SSCMshapepop = \shape
+ o(\|\shape\|_\Fro)$.
\qed

\section{Proofs of Lemma 1 and Proposition 3}\label{app:theoreticalsphericity}
Let $\lambda_1 = \lambdamax(\shape)$. Assume that as $p \to \infty$, $\lambda_1
= O(p^{\tau/2})$, where $\tau<1$. Then 
\begin{equation*}
    \gamma = \frac{\tr(\shape^2)}{p} \leq \frac{p\lambda_1^2}{p} 
    = \lambda_1^2 = O(p^\tau) = o(p).
\end{equation*}
Particularly, $\gamma = o(p)$ if $\lambda_1 = O(1)$.

In the following, we assume that $\varrho$ is a fixed parameter that does not
depend on the dimension $p$. We also use $\gamma = \tr(\shape^2)/p = (1/p)
\sum_{i,j} \Lambda_{ij}^2$, where $\shape = (\Lambda_{ij})$.

\subsection{Sphericity of the AR(1) covariance matrix}\label{app:AR1}
The shape matrix of the AR(1) covariance matrix with parameter $\varrho$
($|\varrho|<1$) has $p$ number of ones on the main diagonal, $2(p-1)$ number of
$\varrho$ on the first diagonals above and below the main diagonal, and $2(p-2)$
number of $\varrho^2$ on the second diagonals above and below the main diagonal,
and so on. That gives,
\begin{equation*}
    \gamma _{\text{AR(1)}}
    = 1 + \frac{2}{p} \sum_{i=1}^{p-1} (p-i) \varrho^{2i}
    = 2\sum_{i=0}^{p-1} (\varrho^{2})^i - \frac{2}{p}\sum_{i=1}^{p-1}
    i(\varrho^{2})^i - 1.
\end{equation*}
The first sum is the geometric series and the second sum is also well-known
and its solution can be obtained by differentiating the geometric series.
Hence, we get
\begin{align*}
    \gamma _{\text{AR(1)}}
        &=
        2 \frac{1-(\varrho^2)^p}{1-\varrho^2}
        - \frac{2}{p} \frac{(p-1)(\varrho^2)^{p+1} - p(\varrho^2)^p +
        \varrho^2}{(1-\varrho^2)^2} - 1
        \\
        &= \frac{p - p \varrho^4 - 2 \varrho^2 + 2 (\varrho^2)^{p+1}}
        {p (\varrho^2-1)^2}.
\end{align*}
As $p \to \infty$, $\gamma_{\text{AR(1)}} \to (1+\varrho^2)/(1-\varrho^2) = O(1)
= o(p)$.

\subsection{Sphericity of the 1-Banded Toeplitz covariance matrix}
The shape matrix corresponding to a 1-banded Toeplitz covariance matrix has $p$
number of ones on the main diagonal and $2(p-1)$ number of $\varrho$ on the
off-diagonals while rest of the off-diagonals are zero. This implies that
$\gamma _{\text{1B}} = 1 + 2(1-1/p) \varrho^2$. As $p \to \infty$, $\gamma
_{\text{1B}} \to 1 + 2 \varrho^2 = O(1) = o(p)$.

\subsection{Sphericity of the spiked covariance matrix}
Let $\M = \M_r + \alpha \I$, where $\M_r$ has rank $r \leq p$. Let $\alpha =
\beta \eta_r$, where $\eta_r = \tr(\M_r)/p$. Let $\gamma_r = p
\tr(\M_r^2)/\tr(\M_r)^2 = \tr(\M_r^2)/(p\eta_r^2)$. Then, $\tr(\M)^2 =
p^2\eta_r^2(1+\beta)^2$ and
\begin{align*}
    \tr(\M^2) 
        &= \tr\big(\big(\M_r + \beta \eta_r \I\big)^2\big)
        =
        \tr(\M_r^2) + \beta^2 p\eta_r^2 + 2\beta p \eta_r^2
        \\
        &= 
        p \eta_r^2 (\gamma_r + \beta^2 + 2\beta)
        = 
        p \eta_r^2
        (\gamma_r - 1 + (\beta + 1)^2).
\end{align*}
Then by computing $\gamma = p \tr(\M^2) / \tr(\M)^2$, we get
\begin{equation*}
    \gamma = \frac{\gamma_r - 1}{(\beta+1)^2} + 1
    = \frac{\gamma_r - 1}{(\alpha/\eta_r + 1)^2} + 1
    \leq
    \frac{\eta_r^2 \gamma_r}{\alpha^2}  + 1 
    \leq \frac{r \lambda_1^2}{p \alpha^2} + 1,
\end{equation*}
where we used that the rank of $\M_r$ is $r$ and $\eta_r^2 \gamma_r =
\tr(\M_r^2) / p \leq r \lambda_1^2/p$, where $\lambda_1 = \lambdamax(\M_r)$.
Therefore, $\gamma = o(p)$, if $\frac{r \lambda_1^2}{p \alpha^2} = o(p)$,
which includes the cases 
\begin{equation*}
    \begin{cases}
        \frac{\lambda_1}{\alpha} = O(p^{\tau/2}),\\
        \frac{\lambda_1}{\alpha} = O(p^{(\tau+1)/2})
        \text{ and } r = O(1),
    \end{cases}
    \quad \text{where}~\tau<1.
\end{equation*}

\subsection{Sphericity of the CS covariance matrix}
The shape matrix of the CS covariance matrix has $p$ number of ones on the
main diagonal and $p(p-1)$ number of $\varrho$ on the off-diagonals. Thus,
$\gamma_{\text{CS}} = 1 + (p-1) \varrho^2 = O(p).$

\qed

\section{Proof of Proposition~\ref{prop:SSCMminusSCM}}\label{app:SSCMminusSCM}
\blue{
    Expanding the mean squared distance gives
    \begin{equation}
        \E[\|\SSCMshape - \shape_{\text{SCM}}\|_\Fro^2]
        =
        \E[\tr(\SSCMshape^2)] + \E[\tr(\shape_{\text{SCM}}^2)]
        - 2\E[\tr(\SSCMshape \shape_{\text{SCM}})].
        \label{eq:SSCMminusSCM}
    \end{equation}
    The first term is
    \begin{align*}
        \E[\tr(\SSCMshape^2)]
        &=
        \E\Big[
            \tr\Big(\Big(
                    \frac{p}{n}\sum_i \frac{\x_i\x_i^\top}{\|\x_i\|^2}
                \Big)^2
        \Big)\Big]
        \\
        &=
        \frac{p^2}{n^2}
        \E\Big[
            \tr\Big(
                \sum_i \frac{\x_i\x_i^\top \x_i\x_i^\top}{\|\x_i\|^4}
                +
                \sum_{i\neq j} \frac{\x_i\x_i^\top \x_j\x_j^\top}
                {\|\x_i\|^2\|\x_j\|^2}
        \Big)\Big]
        \\
        &=
        \frac{p^2}{n} + \frac{n-1}{n}
        \tr(\E[\SSCMshape]^2).
        \numberthis
        \label{eq:term1}
    \end{align*}
    Now, let $\S = \frac{1}{n} \sum_{i=1}^n \x_i\x_i^\top$.  Since $\S$ is
    unbiased, using Wishart theory, we have $\mathrm{MSE}(\S) =
    \tr(\var(\vec(\S))) = (1/n)\tr((\I + \mathbf K)(\M \otimes \M)) =
    (1/n)(\tr(\M)^2 + \tr(\M^2))$, where $\mathbf K = \sum_{i=1}^p\sum_{j=1}^p
    \e_i \e_j^\top \otimes \e_j \e_i^\top$ is the commutation
    matrix~\cite{magnusCommutationMatrixProperties1979} and $\otimes$ denotes
    the Kronecker product. Using that $\E[\tr(\S^2)]  = \mathrm{MSE}(\S) +
    \tr(\M^2)$,
    \begin{align*}
        &\E[\tr(\shape_{\text{SCM}}^2)]
        =
        \frac{p^2}{\tr(\M)^2} \E[\tr(\S^2)]
        =
        \frac{p^2}{n} + \Big(\frac{n+1}{n}\Big) \tr(\shape^2).
        \numberthis
        \label{eq:term2}
    \end{align*}
    The last term of~\eqref{eq:SSCMminusSCM} is
    \begin{align*}
        &\E[\tr(\SSCMshape \shape_{\text{SCM}})]
        =
        \E\Big[ \tr\Big(
                \frac{p}{n}\sum_i \frac{\x_i\x_i^\top}{\|\x_i\|^2} \frac{p}{n\tr(\M)}
                \sum_j \x_j\x_j^\top
        \Big)\Big]
        \\
        &=
        \frac{p^2}{n^2\tr(\M)}\tr\Big( \sum_i \E[\x_i\x_i^\top]
            +
            \sum_{i\neq j} \E\Big[ \frac{\x_i\x_i^\top}{\|\x_i\|^2}\Big]
            \E[\x_j\x_j^\top]
        \Big)
        \\
        &=
        \frac{p^2}{n} + \frac{n-1}{n} \tr(\E[\SSCMshape] \shape).
        \numberthis
        \label{eq:term3}
    \end{align*}
    Now, substituting~\eqref{eq:term1},~\eqref{eq:term2}, and~\eqref{eq:term3}
    into~\eqref{eq:SSCMminusSCM} gives
    \begin{align*}
        \frac{n-1}{n} \big( \tr(\E[\SSCMshape]^2) - 2\tr(\E[\SSCMshape] \shape)
        \big) + \frac{n+1}{n}\tr(\shape^2).
    \end{align*}
    Dividing this by $\|\shape\|_\Fro^2$ and applying assumption $\gamma = o(p)$ and
    Theorem~\ref{thm:sign}, we get the result
    $\E[\|\SSCMshape - \shape_{\text{SCM}}\|_\Fro^2]/\|\shape\|_\Fro^2
    \overset{p\to \infty}\longrightarrow
    \frac{2}{n}$.
}
\qed

\renewcommand*{\bibfont}{\scriptsize}
\printbibliography % when using biblatex

\end{document}

%% file: results/AR1-NMSE-table.tex
	LOOCV & 0.14 {\scriptsize{(0.06)}}& 0.14 {\scriptsize{(0.01)}}& 0.24 {\scriptsize{(0.06)}}& 0.26 {\scriptsize{(0.02)}}& 0.78 {\scriptsize{(0.09)}}\\
	BARTZ & 0.15 {\scriptsize{(0.05)}}& 0.18 {\scriptsize{(0.01)}}& 0.29 {\scriptsize{(0.05)}}& 0.30 {\scriptsize{(0.02)}}& 0.91 {\scriptsize{(0.08)}}\\
	LINPOOL & 0.11 {\scriptsize{(0.05)}}& 0.14 {\scriptsize{(0.01)}}& 0.22 {\scriptsize{(0.05)}}& 0.26 {\scriptsize{(0.02)}}& {\bf{0.73}} {\scriptsize{(0.08)}}\\
	LINPOOL-C & 0.12 {\scriptsize{(0.03)}}& 0.18 {\scriptsize{(0.01)}}& 0.27 {\scriptsize{(0.04)}}& 0.30 {\scriptsize{(0.02)}}& 0.87 {\scriptsize{(0.07)}}

%% file: results/CS-NMSE-table.tex
	LOOCV & 0.18 {\scriptsize{(0.33)}}& 0.05 {\scriptsize{(0.05)}}& 0.18 {\scriptsize{(0.39)}}& 0.06 {\scriptsize{(0.07)}}& 0.47 {\scriptsize{(0.52)}}\\
	BARTZ & 0.20 {\scriptsize{(0.33)}}& 0.05 {\scriptsize{(0.05)}}& 0.13 {\scriptsize{(0.14)}}& 0.06 {\scriptsize{(0.04)}}& 0.43 {\scriptsize{(0.37)}}\\
	LINPOOL & 0.15 {\scriptsize{(0.29)}}& 0.04 {\scriptsize{(0.03)}}& 0.15 {\scriptsize{(0.24)}}& 0.05 {\scriptsize{(0.04)}}& 0.38 {\scriptsize{(0.38)}}\\
	LINPOOL-C & 0.16 {\scriptsize{(0.29)}}& 0.04 {\scriptsize{(0.03)}}& 0.12 {\scriptsize{(0.14)}}& 0.06 {\scriptsize{(0.04)}}& {\bf{0.38}} {\scriptsize{(0.32)}}

%% file: results/MIXED-NMSE-table.tex
	LOOCV & 0.18 {\scriptsize{(0.06)}}& 0.21 {\scriptsize{(0.01)}}& 0.18 {\scriptsize{(0.38)}}& 0.06 {\scriptsize{(0.06)}}& 0.64 {\scriptsize{(0.39)}}\\
	BARTZ & 0.17 {\scriptsize{(0.04)}}& 0.31 {\scriptsize{(0.01)}}& 0.13 {\scriptsize{(0.14)}}& 0.05 {\scriptsize{(0.04)}}& 0.66 {\scriptsize{(0.15)}}\\
	LINPOOL & 0.16 {\scriptsize{(0.05)}}& 0.21 {\scriptsize{(0.01)}}& 0.15 {\scriptsize{(0.24)}}& 0.06 {\scriptsize{(0.05)}}& {\bf{0.58}} {\scriptsize{(0.25)}}\\
	LINPOOL-C & 0.14 {\scriptsize{(0.02)}}& 0.31 {\scriptsize{(0.02)}}& 0.13 {\scriptsize{(0.14)}}& 0.05 {\scriptsize{(0.04)}}& 0.64 {\scriptsize{(0.15)}}

%% file: results/AR1-delta-boxplot.tex
% class1
	\addplot+[boxplot prepared={lower whisker = 2.356231,
		lower quartile = 4.868649,
		median = 6.021367,
		upper quartile = 7.840627,
		upper whisker = 12.261199}]
		table[row sep=\\,y index=0] {12.349339 \\12.363999 \\12.631936 \\12.817819 \\12.865311 \\12.970304 \\13.067132 \\13.246304 \\13.334934 \\13.345886 \\13.408347 \\13.420927 \\13.466393 \\13.477535 \\13.817574 \\13.856300 \\13.897535 \\13.914398 \\14.052705 \\14.478262 \\14.480638 \\14.894254 \\15.851084 \\15.949148 \\16.098431 \\16.174374 \\16.207469 \\16.268642 \\16.788900 \\16.847758 \\17.162573 \\17.312027 \\17.629624 \\18.087631 \\18.366994 \\19.364908 \\20.049070 \\20.589939 \\21.427362 \\21.813634 \\29.873213 \\};
% class2
	\addplot+[boxplot prepared={lower whisker = 3.281343,
		lower quartile = 4.829148,
		median = 5.452510,
		upper quartile = 6.251004,
		upper whisker = 8.382529}]
		table[row sep=\\,y index=0] {8.392392 \\8.544936 \\8.552277 \\8.569457 \\8.582350 \\8.645028 \\8.697477 \\8.718886 \\8.911581 \\8.916490 \\9.037974 \\9.063090 \\9.201853 \\9.342911 \\9.584614 \\9.591336 \\9.672789 \\9.678456 \\9.704579 \\9.905294 \\9.932937 \\10.170266 \\10.277558 \\10.407936 \\10.432581 \\10.444917 \\10.819407 \\10.826479 \\11.143774 \\11.278837 \\11.366372 \\11.703835 \\11.901675 \\12.043591 \\12.626841 \\};
% class3
	\addplot+[boxplot prepared={lower whisker = 20.555601,
		lower quartile = 42.322055,
		median = 54.281050,
		upper quartile = 70.161855,
		upper whisker = 111.698596}]
		table[row sep=\\,y index=0] {113.127360 \\113.749777 \\113.790734 \\113.979339 \\115.584516 \\115.669946 \\116.469929 \\116.945186 \\118.451531 \\118.592782 \\119.525912 \\120.449058 \\120.498072 \\126.229270 \\128.839052 \\129.552511 \\130.968671 \\132.719841 \\133.490266 \\134.441631 \\135.600879 \\136.806989 \\143.950903 \\144.350725 \\146.653598 \\148.300683 \\148.808713 \\154.714782 \\157.406128 \\158.612604 \\159.301052 \\160.049062 \\161.007123 \\161.711849 \\165.357996 \\178.100134 \\178.419746 \\191.670993 \\192.026374 \\194.571301 \\210.690017 \\220.343158 \\290.699865 \\};
% class4
	\addplot+[boxplot prepared={lower whisker = 12.182038,
		lower quartile = 18.994224,
		median = 21.747593,
		upper quartile = 25.420018,
		upper whisker = 34.668237}]
		table[row sep=\\,y index=0] {35.114471 \\35.116996 \\35.260782 \\35.859480 \\35.956425 \\36.765878 \\37.102863 \\37.399175 \\37.628076 \\37.993222 \\38.107947 \\38.413379 \\38.569031 \\38.899124 \\39.065237 \\39.448827 \\39.481352 \\41.156848 \\43.181094 \\43.915216 \\43.980444 \\44.463899 \\47.641156 \\49.778342 \\53.201909 \\53.630896 \\56.503228 \\78.819590 \\79.493660 \\104.913887 \\};
	\addplot[mark=triangle*,very thick,only marks,draw=red,mark size=2pt] coordinates {(7.885867,1) (6.151074,2) (71.399684,3) (24.839397,4)};

%% file: results/CS-delta-boxplot.tex
% class1
	\addplot+[boxplot prepared={lower whisker = 1.770828,
		lower quartile = 4.713983,
		median = 6.205721,
		upper quartile = 8.405752,
		upper whisker = 13.797382}]
		table[row sep=\\,y index=0] {14.010049 \\14.060857 \\14.087997 \\14.098789 \\14.150395 \\14.155047 \\14.426844 \\14.587620 \\14.788669 \\14.796131 \\14.866700 \\14.896743 \\15.227010 \\15.288657 \\15.332098 \\15.520703 \\15.538951 \\15.661379 \\15.663325 \\15.671616 \\15.679797 \\15.936200 \\16.123707 \\16.156229 \\16.302535 \\16.470708 \\16.530039 \\16.741007 \\17.156454 \\17.580369 \\17.818370 \\18.077039 \\18.397633 \\18.407827 \\18.697236 \\18.961809 \\18.982942 \\19.047092 \\19.800146 \\19.951880 \\19.994823 \\20.311498 \\20.426765 \\20.449297 \\20.505568 \\21.447800 \\22.724422 \\23.197244 \\23.401884 \\23.603634 \\24.553104 \\25.152529 \\25.237228 \\25.245949 \\26.476372 \\27.420336 \\32.545236 \\37.593393 \\38.794968 \\40.582308 \\};
% class2
	\addplot+[boxplot prepared={lower whisker = 3.046955,
		lower quartile = 4.873407,
		median = 5.821376,
		upper quartile = 7.157374,
		upper whisker = 10.570146}]
		table[row sep=\\,y index=0] {10.697219 \\10.718873 \\10.752737 \\10.783329 \\10.814819 \\10.817126 \\10.961383 \\11.056157 \\11.180497 \\11.243000 \\11.277110 \\11.393086 \\11.603001 \\11.639348 \\11.700369 \\11.715150 \\11.721156 \\11.741409 \\12.016178 \\12.236771 \\12.332143 \\12.363511 \\12.575522 \\12.822877 \\12.969456 \\13.034083 \\13.492670 \\13.845689 \\13.896637 \\14.101394 \\14.217208 \\14.453524 \\15.265845 \\15.501402 \\16.063282 \\17.495530 \\20.256144 \\21.479468 \\26.815174 \\};
% class3
	\addplot+[boxplot prepared={lower whisker = 11.810761,
		lower quartile = 41.749246,
		median = 57.767977,
		upper quartile = 81.571335,
		upper whisker = 141.137076}]
		table[row sep=\\,y index=0] {142.081884 \\143.276728 \\143.295116 \\143.449693 \\145.034176 \\145.454458 \\145.780609 \\147.059604 \\147.162159 \\147.772871 \\148.313816 \\150.567800 \\151.059839 \\151.464892 \\152.650522 \\154.240055 \\158.201005 \\158.958151 \\159.035714 \\159.212313 \\161.369814 \\162.230969 \\162.352004 \\165.955476 \\167.577380 \\168.929429 \\168.986879 \\170.845419 \\172.467517 \\173.651600 \\173.742923 \\179.410213 \\181.189155 \\181.681432 \\184.505131 \\185.962776 \\186.114684 \\186.331399 \\186.892435 \\188.103883 \\189.222967 \\190.344244 \\190.526461 \\191.130402 \\191.237187 \\195.294799 \\199.633651 \\202.157833 \\205.997289 \\208.383718 \\212.982652 \\213.403407 \\213.654286 \\215.791811 \\217.275393 \\219.205038 \\227.721150 \\230.858636 \\238.274924 \\242.934813 \\246.478211 \\251.018678 \\252.553591 \\255.273751 \\257.085822 \\275.124902 \\278.588260 \\284.185810 \\290.037516 \\308.186690 \\317.682911 \\352.771813 \\358.203438 \\373.224958 \\435.523128 \\488.358777 \\510.306582 \\558.957436 \\683.751754 \\727.561241 \\820.213402 \\};
% class4
	\addplot+[boxplot prepared={lower whisker = 10.358634,
		lower quartile = 19.606890,
		median = 24.611422,
		upper quartile = 32.427640,
		upper whisker = 51.632741}]
		table[row sep=\\,y index=0] {51.751617 \\52.835704 \\53.407773 \\55.148995 \\55.273021 \\55.311721 \\55.987279 \\56.260615 \\56.588468 \\57.178827 \\57.440682 \\57.515124 \\57.526302 \\59.815301 \\60.731928 \\60.833348 \\62.150153 \\64.073832 \\64.200483 \\64.842961 \\65.364708 \\65.557942 \\65.801666 \\65.828358 \\65.869438 \\71.827486 \\72.082949 \\72.167105 \\75.819449 \\76.065860 \\76.627381 \\77.798173 \\79.151893 \\79.350191 \\80.253181 \\82.826209 \\82.946634 \\83.962014 \\89.942034 \\91.268297 \\106.315604 \\108.778683 \\114.422489 \\129.753607 \\301.240960 \\307.673748 \\};
	\addplot[mark=triangle*,very thick,only marks,draw=red,mark size=2pt] coordinates {(8.780237,1) (7.394408,2) (93.653289,3) (35.945632,4)};

%% file: results/AR1-cij-boxplot.tex
% 11
	\addplot+[boxplot prepared={lower whisker = 0.437640,
		lower quartile = 0.922781,
		median = 1.129096,
		upper quartile = 1.434873,
		upper whisker = 2.195899}]
		table[row sep=\\,y index=0] {2.232595 \\2.238587 \\2.265545 \\2.281374 \\2.290987 \\2.335806 \\2.342622 \\2.348590 \\2.357992 \\2.365062 \\2.378938 \\2.391728 \\2.439863 \\2.517211 \\2.561132 \\2.562199 \\2.639418 \\2.667697 \\2.694314 \\2.722499 \\2.779117 \\2.953839 \\3.000351 \\3.037017 \\3.048476 \\3.084675 \\3.087392 \\3.304237 \\3.437338 \\3.942604 \\};
% 12
	\addplot+[boxplot prepared={lower whisker = 1.492970,
		lower quartile = 2.197876,
		median = 2.472515,
		upper quartile = 2.758281,
		upper whisker = 3.596057}]
		table[row sep=\\,y index=0] {3.618818 \\3.639658 \\3.660285 \\3.663571 \\3.698954 \\3.716634 \\3.729275 \\3.753485 \\3.791979 \\3.795791 \\3.816338 \\3.906642 \\3.955184 \\4.054798 \\4.151234 \\4.163791 \\4.183773 \\4.255564 \\4.260589 \\4.278295 \\4.480132 \\};
% 13
	\addplot+[boxplot prepared={lower whisker = 1.826764,
		lower quartile = 3.281456,
		median = 3.809009,
		upper quartile = 4.521051,
		upper whisker = 6.359717}]
		table[row sep=\\,y index=0] {6.457620 \\6.522141 \\6.552048 \\6.559566 \\6.561051 \\6.591937 \\6.606908 \\6.643069 \\6.663339 \\6.713349 \\6.760386 \\6.760750 \\6.891349 \\6.968439 \\7.016751 \\7.270841 \\7.317054 \\7.361476 \\7.564302 \\7.966051 \\8.271661 \\};
% 14
	\addplot+[boxplot prepared={lower whisker = 3.205796,
		lower quartile = 4.884484,
		median = 5.516304,
		upper quartile = 6.231184,
		upper whisker = 8.231286}]
		table[row sep=\\,y index=0] {8.256106 \\8.284756 \\8.365093 \\8.415515 \\8.535671 \\8.561670 \\8.623633 \\8.656196 \\8.812339 \\8.867433 \\8.926533 \\8.954284 \\8.982644 \\9.066095 \\9.170922 \\9.222350 \\9.252225 \\9.275468 \\9.737913 \\10.174409 \\10.391597 \\10.598935 \\};
% 22
	\addplot+[boxplot prepared={lower whisker = 3.512659,
		lower quartile = 4.875956,
		median = 5.351845,
		upper quartile = 5.931626,
		upper whisker = 7.509346}]
		table[row sep=\\,y index=0] {7.581017 \\7.603737 \\7.613266 \\7.613446 \\7.670452 \\7.687964 \\7.707591 \\7.720589 \\7.759965 \\7.772128 \\7.784638 \\7.820820 \\7.924341 \\8.120546 \\8.175566 \\8.315052 \\8.373665 \\8.960402 \\};
% 23
	\addplot+[boxplot prepared={lower whisker = 5.204352,
		lower quartile = 7.599265,
		median = 8.566836,
		upper quartile = 9.631379,
		upper whisker = 12.654117}]
		table[row sep=\\,y index=0] {12.695725 \\12.757186 \\12.821788 \\12.849140 \\12.863894 \\12.888082 \\12.901742 \\12.966188 \\12.969260 \\13.079928 \\13.178390 \\13.246265 \\13.256175 \\13.262665 \\13.420547 \\13.565198 \\13.663463 \\13.721090 \\13.780048 \\13.782201 \\13.810569 \\13.948267 \\14.270095 \\14.552788 \\14.586777 \\14.809029 \\14.842564 \\14.994048 \\16.383662 \\};
% 24
	\addplot+[boxplot prepared={lower whisker = 9.214135,
		lower quartile = 11.704037,
		median = 12.533765,
		upper quartile = 13.468095,
		upper whisker = 16.085477}]
		table[row sep=\\,y index=0] {16.164685 \\16.282060 \\16.401027 \\16.460495 \\16.472258 \\16.530387 \\16.677691 \\16.961932 \\17.081300 \\17.224068 \\17.440567 \\17.562320 \\17.615749 \\17.850636 \\17.868478 \\};
% 33
	\addplot+[boxplot prepared={lower whisker = 5.255029,
		lower quartile = 10.985040,
		median = 13.645841,
		upper quartile = 17.445762,
		upper whisker = 27.035871}]
		table[row sep=\\,y index=0] {27.335854 \\27.474424 \\27.617620 \\28.094443 \\28.235343 \\28.793533 \\28.999361 \\29.014908 \\29.037852 \\29.572422 \\29.584748 \\30.625066 \\31.363516 \\31.643309 \\31.968084 \\32.070401 \\32.220857 \\33.254719 \\33.513809 \\33.753902 \\34.015912 \\34.380069 \\34.817820 \\35.143138 \\35.962065 \\35.969779 \\36.542481 \\37.936216 \\45.404427 \\46.374080 \\47.970564 \\53.019526 \\};
% 34
	\addplot+[boxplot prepared={lower whisker = 12.052185,
		lower quartile = 18.395411,
		median = 20.683089,
		upper quartile = 23.604873,
		upper whisker = 31.386428}]
		table[row sep=\\,y index=0] {32.115619 \\32.424261 \\32.564792 \\32.699006 \\32.793512 \\32.865199 \\33.007145 \\33.049949 \\33.461741 \\33.718114 \\33.742518 \\34.197950 \\34.231510 \\34.410452 \\34.518778 \\34.803473 \\34.969257 \\34.976006 \\35.567518 \\36.072825 \\38.126507 \\38.918980 \\39.789801 \\};
% 44
	\addplot+[boxplot prepared={lower whisker = 21.227213,
		lower quartile = 29.082469,
		median = 32.013954,
		upper quartile = 35.278475,
		upper whisker = 44.565343}]
		table[row sep=\\,y index=0] {18.415985 \\44.584447 \\44.816280 \\44.905747 \\45.083069 \\45.457151 \\45.879961 \\45.928813 \\46.736233 \\46.747123 \\46.955014 \\47.040936 \\47.284279 \\48.484492 \\48.960987 \\49.042409 \\50.106120 \\50.717513 \\53.228107 \\54.486752 \\65.128962 \\};
	\addplot[mark=triangle*,very thick,only marks,draw=red,mark size=2pt] coordinates {(1.195629,1) (2.539256,2) (4.046367,3) (5.734682,4) (5.505669,5) (8.962500,6) (12.986150,7) (14.920000,8) (22.138776,9) (33.718750,10)};

%% file: results/CS-cij-boxplot.tex
% 11
	\addplot+[boxplot prepared={lower whisker = 0.772387,
		lower quartile = 3.278713,
		median = 5.085038,
		upper quartile = 7.964300,
		upper whisker = 14.983591}]
		table[row sep=\\,y index=0] {15.004610 \\15.233398 \\15.613231 \\15.664043 \\15.810065 \\15.844552 \\15.949275 \\16.059304 \\16.083688 \\16.093287 \\16.278071 \\16.452751 \\16.548551 \\16.591827 \\16.675125 \\16.735809 \\16.861883 \\17.135508 \\17.224958 \\17.542492 \\17.980484 \\17.987901 \\18.328037 \\18.612362 \\19.044516 \\19.136620 \\19.352006 \\19.469049 \\21.717084 \\22.327489 \\22.913180 \\23.148577 \\23.505035 \\23.595844 \\23.604061 \\24.733489 \\25.973683 \\27.333814 \\27.437322 \\28.627125 \\30.280366 \\36.531701 \\37.800975 \\54.636744 \\};
% 12
	\addplot+[boxplot prepared={lower whisker = 3.811834,
		lower quartile = 10.663850,
		median = 13.354235,
		upper quartile = 17.238625,
		upper whisker = 26.973661}]
		table[row sep=\\,y index=0] {27.380757 \\27.396524 \\27.592516 \\27.612129 \\27.635235 \\28.133686 \\28.338817 \\28.805685 \\28.881765 \\28.932586 \\29.111595 \\29.133207 \\29.378109 \\30.312525 \\30.463741 \\30.561642 \\34.273705 \\34.278792 \\35.791345 \\41.006777 \\41.055961 \\45.585336 \\};
% 13
	\addplot+[boxplot prepared={lower whisker = 5.377110,
		lower quartile = 16.557649,
		median = 22.811403,
		upper quartile = 31.205722,
		upper whisker = 53.083490}]
		table[row sep=\\,y index=0] {53.425393 \\53.431162 \\54.115304 \\55.195632 \\55.379321 \\55.484118 \\55.750666 \\55.757469 \\55.900479 \\56.902450 \\57.083423 \\57.089600 \\57.296478 \\57.519965 \\57.620392 \\57.845229 \\58.502483 \\58.526805 \\58.642989 \\58.846481 \\58.858703 \\58.956449 \\59.687968 \\59.877978 \\60.137722 \\60.247979 \\61.772275 \\61.833776 \\62.990152 \\63.582680 \\64.108440 \\65.016702 \\66.298618 \\67.049669 \\70.541418 \\70.839875 \\71.243764 \\71.665263 \\75.328977 \\83.702921 \\99.956696 \\104.903876 \\117.354914 \\};
% 14
	\addplot+[boxplot prepared={lower whisker = 9.782534,
		lower quartile = 28.727503,
		median = 37.301621,
		upper quartile = 47.490923,
		upper whisker = 73.633614}]
		table[row sep=\\,y index=0] {75.660587 \\75.799947 \\76.494961 \\76.832143 \\76.866256 \\78.890070 \\79.945098 \\80.803290 \\86.397489 \\86.815777 \\87.231539 \\88.122304 \\88.206889 \\88.836683 \\89.587466 \\91.667596 \\93.498059 \\98.651048 \\101.263134 \\101.736082 \\102.479537 \\103.855588 \\107.145095 \\113.770082 \\143.846646 \\153.781934 \\};
% 22
	\addplot+[boxplot prepared={lower whisker = 13.597621,
		lower quartile = 27.935686,
		median = 34.201230,
		upper quartile = 42.034893,
		upper whisker = 62.557594}]
		table[row sep=\\,y index=0] {63.526894 \\64.569601 \\64.706286 \\65.898832 \\65.954011 \\66.060084 \\67.263759 \\68.091651 \\68.464306 \\69.111891 \\69.693551 \\69.747987 \\70.394734 \\71.057398 \\72.479100 \\73.053337 \\78.104848 \\79.884290 \\};
% 23
	\addplot+[boxplot prepared={lower whisker = 14.942032,
		lower quartile = 46.102561,
		median = 59.222479,
		upper quartile = 75.546130,
		upper whisker = 119.320011}]
		table[row sep=\\,y index=0] {119.836014 \\120.018654 \\120.787092 \\122.046563 \\122.511120 \\123.031069 \\123.556355 \\124.917691 \\125.139643 \\127.137986 \\127.606984 \\127.898876 \\127.940239 \\129.719170 \\130.086613 \\130.958917 \\132.436129 \\134.314955 \\136.306995 \\136.913296 \\137.272491 \\138.804540 \\139.610004 \\141.654744 \\145.440677 \\148.159303 \\148.830560 \\149.508075 \\151.534822 \\151.604187 \\153.433564 \\156.015160 \\157.509993 \\161.081384 \\162.213500 \\173.942085 \\184.937390 \\193.365973 \\};
% 24
	\addplot+[boxplot prepared={lower whisker = 53.314385,
		lower quartile = 82.109998,
		median = 95.751912,
		upper quartile = 111.889399,
		upper whisker = 155.990697}]
		table[row sep=\\,y index=0] {158.000100 \\159.421227 \\159.499366 \\161.074253 \\163.608700 \\165.346938 \\165.359045 \\166.012720 \\166.217682 \\166.598673 \\167.749564 \\168.842665 \\176.035731 \\177.429315 \\178.154623 \\178.445179 \\188.866715 \\};
% 33
	\addplot+[boxplot prepared={lower whisker = 6.904398,
		lower quartile = 60.966375,
		median = 100.010213,
		upper quartile = 153.420534,
		upper whisker = 291.213548}]
		table[row sep=\\,y index=0] {292.494622 \\297.274060 \\299.125310 \\299.141608 \\299.696961 \\300.684026 \\301.443849 \\308.761982 \\309.474612 \\310.001352 \\313.279473 \\314.192043 \\315.855316 \\320.423514 \\325.209436 \\332.856707 \\336.369789 \\340.974796 \\341.613606 \\344.107543 \\348.884991 \\350.619742 \\350.992112 \\352.311481 \\353.025933 \\354.920754 \\362.481423 \\367.455152 \\371.490941 \\373.121551 \\378.441610 \\380.993267 \\382.596256 \\388.361776 \\411.176854 \\414.883678 \\415.879184 \\424.646961 \\426.079513 \\427.012934 \\435.425114 \\435.559724 \\436.558642 \\439.502699 \\440.738016 \\441.198884 \\446.950117 \\457.524146 \\461.701280 \\465.017255 \\468.103005 \\475.592200 \\487.360288 \\494.550255 \\498.419437 \\546.404166 \\552.235245 \\589.109160 \\653.430995 \\725.276684 \\733.734904 \\778.211521 \\821.464456 \\941.296928 \\};
% 34
	\addplot+[boxplot prepared={lower whisker = 35.363676,
		lower quartile = 129.276394,
		median = 164.881410,
		upper quartile = 214.357762,
		upper whisker = 339.814789}]
		table[row sep=\\,y index=0] {341.999166 \\344.693778 \\349.614709 \\349.968715 \\350.034175 \\355.478639 \\362.386555 \\362.608403 \\362.826712 \\365.630021 \\367.176694 \\367.194490 \\367.245377 \\370.041553 \\372.422404 \\373.653293 \\374.902263 \\376.043104 \\379.178925 \\388.078409 \\389.218793 \\391.940461 \\400.308250 \\401.089910 \\408.651498 \\411.483118 \\412.456427 \\418.398860 \\421.443910 \\438.358626 \\461.019102 \\464.438915 \\496.521358 \\571.426425 \\607.748911 \\};
% 44
	\addplot+[boxplot prepared={lower whisker = 98.179689,
		lower quartile = 211.294057,
		median = 266.502373,
		upper quartile = 333.418448,
		upper whisker = 512.545534}]
		table[row sep=\\,y index=0] {519.410684 \\521.212809 \\522.553230 \\523.504411 \\526.667353 \\527.330623 \\536.640120 \\538.156193 \\538.527241 \\542.568291 \\547.425993 \\548.470461 \\548.957879 \\552.509810 \\552.639138 \\553.826946 \\556.776676 \\560.627363 \\562.054290 \\565.685036 \\568.457713 \\568.821503 \\569.729198 \\573.280713 \\575.737081 \\580.425363 \\587.238825 \\589.556462 \\597.325255 \\597.818893 \\609.438513 \\613.655571 \\700.463090 \\748.650243 \\778.809703 \\917.831304 \\924.885249 \\};
	\addplot[mark=triangle*,very thick,only marks,draw=red,mark size=2pt] coordinates {(9.910000,1) (25.760000,2) (47.550000,3) (75.280000,4) (67.360000,5) (124.800000,6) (198.080000,7) (231.750000,8) (368.400000,9) (586.240000,10)};

%% file: results/AR1-LIN2-a4-boxplot.tex
% a14
	\addplot+[boxplot prepared={lower whisker = 0.035687,
		lower quartile = 0.099578,
		median = 0.119717,
		upper quartile = 0.144669,
		upper whisker = 0.210975}]
		table[row sep=\\,y index=0] {0.212612 \\0.216054 \\0.217856 \\0.224744 \\0.227791 \\0.230770 \\0.245556 \\0.251020 \\0.256417 \\};
% a24
	\addplot+[boxplot prepared={lower whisker = 0.219914,
		lower quartile = 0.352560,
		median = 0.397734,
		upper quartile = 0.447200,
		upper whisker = 0.588024}]
		table[row sep=\\,y index=0] {0.179412 \\0.195057 \\0.201972 \\0.590243 \\0.592483 \\0.596456 \\0.601568 \\0.609748 \\0.615527 \\0.626324 \\0.628644 \\0.712645 \\0.719314 \\};
% a34
	\addplot+[boxplot prepared={lower whisker = 0.032324,
		lower quartile = 0.067177,
		median = 0.079732,
		upper quartile = 0.095674,
		upper whisker = 0.138290}]
		table[row sep=\\,y index=0] {0.023955 \\0.138859 \\0.141391 \\0.144145 \\0.144457 \\0.151093 \\0.151469 \\0.154004 \\0.156190 \\0.161526 \\0.164728 \\};
% a44
	\addplot+[boxplot prepared={lower whisker = 0.280752,
		lower quartile = 0.334428,
		median = 0.352755,
		upper quartile = 0.370749,
		upper whisker = 0.424770}]
		table[row sep=\\,y index=0] {0.154819 \\0.192529 \\0.233133 \\0.239509 \\0.243844 \\0.244608 \\0.253248 \\0.255202 \\0.258960 \\0.259562 \\0.266038 \\0.266270 \\0.269834 \\0.270366 \\0.271756 \\0.273092 \\0.273105 \\0.273678 \\0.274305 \\0.274820 \\0.276398 \\0.278912 \\0.279072 \\0.279866 \\0.426268 \\0.428567 \\};
% a54
	\addplot+[boxplot prepared={lower whisker = 1.046664,
		lower quartile = 1.326503,
		median = 1.421703,
		upper quartile = 1.520518,
		upper whisker = 1.809665}]
		table[row sep=\\,y index=0] {1.815268 \\1.817285 \\1.848539 \\1.860520 \\1.869569 \\1.878109 \\1.901831 \\1.912677 \\1.919114 \\1.944248 \\1.944627 \\1.978072 \\2.009743 \\2.032194 \\2.049334 \\2.052052 \\2.075647 \\2.330512 \\2.422442 \\2.478184 \\};
	\addplot[mark=triangle*,very thick,only marks,draw=red,mark size=2pt] coordinates {(0.104616,1) (0.389221,2) (0.068896,3) (0.350063,4) (1.510002,5)};

%% file: results/CS-LIN2-a4-boxplot.tex
% a14
	\addplot+[boxplot prepared={lower whisker = 0.034339,
		lower quartile = 0.201706,
		median = 0.260905,
		upper quartile = 0.333599,
		upper whisker = 0.528947}]
		table[row sep=\\,y index=0] {0.532786 \\0.558355 \\0.565106 \\0.580944 \\0.590429 \\0.591859 \\0.605861 \\0.612021 \\0.621831 \\0.634451 \\0.669772 \\0.670419 \\0.684742 \\0.699476 \\0.705574 \\0.823771 \\};
% a24
	\addplot+[boxplot prepared={lower whisker = 0.285926,
		lower quartile = 0.673874,
		median = 0.796945,
		upper quartile = 0.947617,
		upper whisker = 1.352955}]
		table[row sep=\\,y index=0] {1.358490 \\1.389726 \\1.393087 \\1.393449 \\1.410382 \\1.424358 \\1.435585 \\1.461947 \\1.466896 \\1.474197 \\1.477359 \\1.485417 \\1.492655 \\1.497602 \\1.537029 \\1.590654 \\1.631025 \\1.755282 \\1.807612 \\2.479973 \\2.764337 \\};
% a34
	\addplot+[boxplot prepared={lower whisker = 0.022304,
		lower quartile = 0.115217,
		median = 0.148163,
		upper quartile = 0.188183,
		upper whisker = 0.295446}]
		table[row sep=\\,y index=0] {0.300215 \\0.302516 \\0.304429 \\0.307456 \\0.309037 \\0.314563 \\0.314688 \\0.316984 \\0.322087 \\0.323452 \\0.333250 \\0.335596 \\0.356619 \\0.367033 \\0.377702 \\0.493364 \\0.493432 \\};
% a44
	\addplot+[boxplot prepared={lower whisker = 0.335952,
		lower quartile = 0.476699,
		median = 0.524250,
		upper quartile = 0.571516,
		upper whisker = 0.698712}]
		table[row sep=\\,y index=0] {0.207345 \\0.224939 \\0.234877 \\0.240557 \\0.250660 \\0.267489 \\0.268445 \\0.273649 \\0.279834 \\0.284673 \\0.296919 \\0.302259 \\0.302444 \\0.308883 \\0.314030 \\0.314772 \\0.328212 \\0.328264 \\0.715384 \\0.719616 \\};
% a54
	\addplot+[boxplot prepared={lower whisker = 0.000000,
		lower quartile = 0.000000,
		median = 0.000000,
		upper quartile = 0.000000,
		upper whisker = 0.000000}]
		table[row sep=\\,y index=0] {0.000000 \\0.000000 \\0.000000 \\0.000000 \\0.000000 \\0.000000 \\0.000000 \\0.000000 \\0.000000 \\0.000000 \\0.000000 \\0.000000 \\0.000000 \\0.000000 \\0.000000 \\0.000000 \\0.000000 \\0.000000 \\0.000000 \\0.000000 \\0.000000 \\0.000000 \\0.000000 \\0.000000 \\0.000000 \\0.000000 \\0.000000 \\0.000000 \\0.000000 \\0.000000 \\0.000000 \\0.000000 \\0.000000 \\0.000000 \\0.000000 \\0.000000 \\0.000000 \\0.000000 \\0.000000 \\0.000000 \\0.000000 \\0.000000 \\0.000000 \\0.000000 \\0.000000 \\0.000000 \\0.000000 \\0.000000 \\0.000000 \\0.000000 \\0.000000 \\0.000000 \\0.000000 \\0.000000 \\0.000000 \\0.000000 \\0.000000 \\0.000000 \\0.000000 \\0.000000 \\0.000000 \\0.000000 \\0.000000 \\0.000000 \\0.000000 \\0.000000 \\0.000000 \\0.000000 \\0.000000 \\0.000000 \\0.000000 \\0.000000 \\0.000000 \\0.000000 \\0.000000 \\0.000000 \\0.000000 \\0.000000 \\0.000000 \\0.000000 \\0.000000 \\0.000000 \\0.000000 \\0.000000 \\0.000000 \\0.000000 \\0.000000 \\0.000000 \\0.000000 \\0.000000 \\0.000000 \\0.000000 \\0.000000 \\0.000000 \\0.000000 \\0.000000 \\0.000000 \\0.000000 \\0.000000 \\0.000000 \\0.000000 \\0.000000 \\0.000000 \\0.000000 \\0.000000 \\0.000000 \\0.000000 \\0.000000 \\0.000000 \\0.000000 \\0.000000 \\0.000000 \\0.000000 \\0.000000 \\0.000000 \\0.000000 \\0.000000 \\0.000000 \\0.000000 \\0.000000 \\0.000000 \\0.000000 \\0.000000 \\0.000000 \\0.000000 \\0.000000 \\0.000000 \\0.000000 \\0.000000 \\0.000000 \\0.000000 \\0.000000 \\0.000000 \\0.000000 \\0.000000 \\0.000000 \\0.000000 \\0.000000 \\0.000000 \\0.000001 \\0.004106 \\0.008883 \\0.009889 \\0.026319 \\0.028299 \\0.035946 \\0.046038 \\0.059686 \\0.078931 \\0.109395 \\0.121995 \\0.185976 \\};
	\addplot[mark=triangle*,very thick,only marks,draw=red,mark size=2pt] coordinates {(0.252936,1) (0.855356,2) (0.131464,3) (0.561470,4) (0.000000,5)};